%% file: ms.tex
\documentclass[acmsmall,screen]{acmart}
\newtheorem{theorem}{Theorem}
\newtheorem{definition}{Definition}

\usepackage{amsmath, bm}
\usepackage{algorithm}
\usepackage{booktabs}
\usepackage{algorithmic}
\usepackage{subfigure}
\usepackage{multirow}
\usepackage[switch]{lineno}

\AtBeginDocument{%
  \providecommand\BibTeX{{%
    \normalfont B\kern-0.5em{\scshape i\kern-0.25em b}\kern-0.8em\TeX}}}

\setcopyright{acmcopyright}
\copyrightyear{2021}
\acmYear{2021}
\acmDOI{10.1145/1122445.1122456}

\acmJournal{JACM}
\acmVolume{37}
\acmNumber{4}
\acmArticle{111}
\acmMonth{5}

\begin{document}

\title{Criterion-based Heterogeneous Collaborative Filtering for Multi-behavior Implicit Recommendation}

\author{Xiao Luo}
\orcid{0000-0002-7987-3714}
\authornote{Xiao Luo, Daqing Wu and Yiyang Gu contributed equally to this research. }
\email{xiaoluo@cs.ucla.edu}
\affiliation{%
  \institution{Department of Computer Science, University of California, Los Angeles}
  \country{USA}
  \postcode{90024}
}

\author{Daqing Wu}
\orcid{0000-0002-3354-0973}
\authornotemark[1]
\email{wudq@pku.edu.cn}
\affiliation{%
  \institution{School of Mathematical Sciences, Peking University}
  \city{Beijing}
  \country{China}
  \postcode{100871}
}
\author{Yiyang Gu}
\orcid{0000-0002-5915-4448}
\authornotemark[1]
\email{yiyanggu@pku.edu.cn}
\affiliation{%
  \institution{National Key Laboratory for Multimedia Information Processing, School of Computer Science, Peking University}
  \city{Beijing}
  \country{China}
  \postcode{100871}
}

\author{Chong Chen}
\orcid{0000-0003-0213-9957}
\affiliation{%
  \institution{Terminus Group}
  \city{Beijing}
  \country{China}}
  \postcode{100027}
\email{cheung1990@126.com}
\authornote{Corresponding authors.}

\author{Luchen Liu}
\email{liuluchen@pku.edu.cn}
\orcid{0009-0009-3817-6267}
\affiliation{%
  \institution{National Key Laboratory for Multimedia Information Processing, School of Computer Science, Peking University}
  \city{Beijing}
  \country{China}
  \postcode{100871}
}
\author{Jinwen Ma}
\orcid{0000-0002-7388-4295}
\email{jwma@math.pku.edu.cn}
\affiliation{%
  \institution{School of Mathematical Sciences, Peking University}
  \city{Beijing}
  \country{China}
  \postcode{100871}
}
\author{Ming Zhang}
\orcid{0000-0002-9809-3430}
\email{mzhang_cs@pku.edu.cn}
\affiliation{%
  \institution{National Key Laboratory for Multimedia Information Processing, School of Computer Science, Peking University}
  \city{Beijing}
  \country{China}
  \postcode{100871}
}
\author{Minghua Deng}
\orcid{0000-0002-9143-1898}
\email{dengmh@pku.edu.cn}
\affiliation{%
  \institution{School of Mathematical Sciences, Peking University}
  \city{Beijing}
  \country{China}
  \postcode{100871}
}
\authornotemark[2]
\author{Jianqiang Huang}
\orcid{0000-0001-5735-2910}
\email{jianqiang.jqh@gmail.com}
\affiliation{%
  \institution{Nanyang Technological University}
  \country{Singapore}
  \postcode{639798}
}

\author{Xian-Sheng Hua}
\orcid{0000-0002-8232-5049}
\email{huaxiansheng@gmail.com}
\affiliation{%
  \institution{Terminus Group}
  \city{Beijing}
  \country{China}
  \postcode{100027}
}

\renewcommand{\shortauthors}{Luo et al.}
\newcommand{\R}[1]{{\color{black}{#1}}}
\newcommand{\RR}[1]{{\color{black}{#1}}}
\def\method{CHCF}

\input{abstract}

\ccsdesc[500]{Information systems ~ Recommender systems}
\ccsdesc[300]{Computing methodologies ~ Neural networks}

\keywords{Collaborative filtering, neural networks, implicit feedback, multi-behavior recommendation}

\maketitle

\input{introduction}

\input{related_work}

\input{method}
\input{experiment}

\input{conclusion}


\begin{acks}
The authors are grateful to the anonymous reviewers for critically reading the manuscript and for giving important suggestions to improve their paper. This work was supported by the National Key Research and Development
Program of China (2021YFF1200902) and the National Natural Science
Foundation of China (32270689, 62106008 and 62276002).

\end{acks}

\bibliographystyle{ACM-Reference-Format}
\bibliography{ms}


\
\end{document}

%% file: abstract.tex
\begin{abstract}

Recent years have witnessed the explosive growth of interaction behaviors in multimedia information systems, where multi-behavior recommender systems have received increasing attention by leveraging data from various auxiliary behaviors such as tip and collect.
Among various multi-behavior recommendation methods, non-sampling methods have shown superiority over negative sampling methods. However, two observations are usually ignored in existing state-of-the-art non-sampling methods based on binary regression: (1) users have different preference strengths for different items, so they cannot be measured simply by binary implicit data; (2) the dependency across multiple behaviors varies for different users and items.
To tackle the above issue, we propose a novel non-sampling learning framework named \underline{C}riterion-guided \underline{H}eterogeneous \underline{C}ollaborative \underline{F}iltering (CHCF). CHCF introduces both upper and lower thresholds to indicate selection criteria, which will guide user preference learning. Besides, CHCF integrates criterion learning and user preference learning into a unified framework, which can be trained jointly for the interaction prediction of the target behavior. We further theoretically demonstrate that the optimization of Collaborative Metric Learning can be approximately achieved by the CHCF learning framework in a non-sampling form effectively. Extensive experiments on three real-world datasets show the effectiveness of CHCF in heterogeneous scenarios.

\end{abstract}

%% file: introduction.tex
\section{Introduction}

Recommender systems (RS) have received extensive attention in information systems in recent years \cite{ricci2011introduction,zhou2019ijcai,chen2020learning,wang2022learning,ju2023comprehensive,qin2023disenpoi}. The key point of RS is to depict user-item interaction with the exploration of user preferences and item intrinsic properties~\cite{he2020lightgcn}. As a prevalent and basic procedure for constructing a recommender system, Collaborative filtering (CF)~\cite{hu2008collaborative} presumes that users with comparable behavioral characteristics have similar preferences on various items.
In order to accomplish this assumption, CF learns parameterized user and item representation and estimates user preferences from the historical interaction data. Generally, there are two core components to learn the CF model, i.e., the base component is user and item embedding, and another is interaction modeling where  historical interactions are reconstructed according to the user and item embeddings. Among the extensive CF methods, matrix factorization (MF) \cite{hu2008collaborative, koren2009matrix,liu2021bayesian} is popular and widely used in many applications \cite{liu2017related}. Early works on MF algorithms mostly deal with explicit feedback, in which the ratings of users directly imply the preference on corresponding items \cite{koren2008factorization}. However, explicit ratings are difficult to obtain in some cases. In most cases, our data contains implicit feedback~\cite{yang2017improving}, e.g., collecting and viewing records. Owing to the success of deep learning in vision and language fields \cite{sang2020context}, there are also several attempts which adopt neural networks (NN) in the field of recommender systems \cite{deng2019deepcf,he2020lightgcn,zhang2021hybrid}. These NN-based methods learn low-dimensional user and item embeddings, followed by predicting interaction based on the embedding vectors based on various network architectures.

Despite the predominance of the aforementioned CF models, they merely investigate a single form of user-item interactions when learning user preference representations. Nevertheless, in real-world practice, user-item interactions can be multiplex and characterized by a range of relationships~\cite{gao2019learning,xia2020multiplex}. In E-commerce platforms, for instance, users exhibit different types of behaviors including view, add-to-cart, add-to-favorite and purchase. Due to the idiosyncrasies of individual users, the dependencies across multiple behaviors are various. For instance, some users prefer add-to-cart before purchase, which indicates that add-to-cart behavior probably co-occur along with purchase behavior. But, some users are more likely to purchase items from the add-to-favorite item lists. These anfractuous dependencies across multiple behaviors increase the difficulty of distilling effective collaborative signals from multi-behavior interactions. The early solutions to the above multi-behavior recommendation are on the basis of collective matrix factorization \cite{zhao2015improving, singh2008relational,krohn2012multi}. These methods extend MF to jointly factorize multiple matrices. To tackle the dilemma in modeling the dependencies across multiple behaviors, a handful of studies consider the diverse behavior relationships among multiple types of interactions~\cite{gao2019neural,chen2020efficient}. For example, NMTR~\cite{gao2019neural} presumes that the behavior categories have cascading relationships and transmit interaction information on this basis. EHCF \cite{chen2020efficient} models the relationships among multiple behaviors using a transfer model. As we implied before, the dependencies across multiple behaviors are various from users or items. The above methods do not take the behavior dependencies into consideration or merely account for it in a simple assumption which is independent of the users/items. Recently, Jin et al. \cite{jin2020multi} tackle this problem by constructing a unified graph and modeling the dependencies across multiple behaviors by parameterized user-specific behavior importance weight. It happens that there is a similar case. GHCF \cite{chen2021graph} jointly embeds the representations of users and items with the graph structure, along with relations for multi-behavior prediction and trains the model with non-sampling optimization, achieving state-of-the-art performance. In summary, the methods with non-sampling learning strategies \cite{chen2020efficient,chen2021graph} typically perform better than sampling-based strategies for the multi-behavior recommendation task.

\begin{figure}
\centering
\includegraphics[width=0.75\textwidth,keepaspectratio=true]{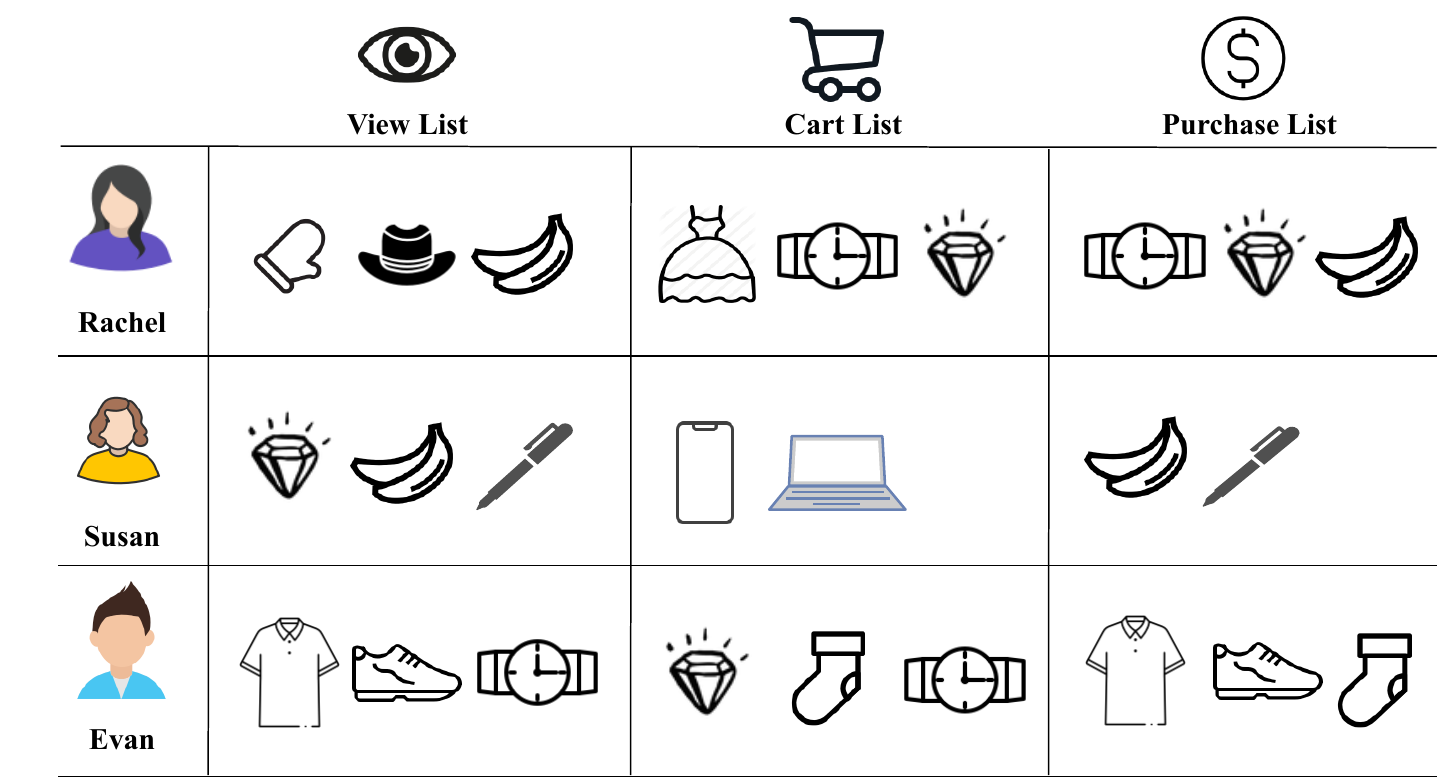}
\caption{The strength of multiple types of behavior is not consistent for different users and items because of user- and item-specific selection criteria. For example, Rachel tends to buy items that occur in her cart list compared with Susan, which may result from different consumption views. Bananas are more likely to be purchased after being viewed than diamonds. The underlying reason is that diamonds belong to luxury goods.}
\label{fig:challenges}
\end{figure}

Although the existing heterogeneous collaborative filtering methods achieve satisfactory results, they still face two key challenges: (1) Different preference strengths cannot be measured by implicit data. Specifically, an observed interaction record (labeled as 1) indicates an interaction between the user and the item, but it doesn't promise that the user really likes it. Similarly, an unobserved interaction record (labeled as 0) doesn't necessarily stand for that the user doesn't like the item; it could be that the user hasn't been aware of this item. Therefore, recent non-sampling strategies \cite{chen2020efficient,chen2021graph} regress each observed/unobserved interaction score to 1/0 under each behavior, which neglects the rating caused by implicit feedback and further hurts diverse user preference learning. 
(2) The dependencies across multiple types of behavior are various for different users and items, since the selection criteria of users for different behaviors are various due to user idiosyncrasy (e.g., consumption ability) and item properties (e.g., prices).
In other words, if the selection criteria of two behaviors are similar (dissimilar) for a user-item pair, they have a close (distant) relation. For example, in Fig. \ref{fig:challenges}, Rachel is more likely to buy items that occur in her cart list than Susan. Thus carting implies a stronger collaborative signal for Rachel rather than Susan. Furthermore, it requires more considerable thought to purchase luxuries (e.g., diamonds) than daily necessities (e.g., bananas), even if the user may often visit luxury stores. Thus the view behaviors indicate a stronger collaborative signal for daily necessities than luxuries. However, most of the recommendation models fail to capture this heterogeneity comprehensively.
For example, MRMN \cite{zhou2019ijcai} adopts the collaborative metric learning (CML) \cite{hsieh2017collaborative} framework, involving a triple loss paradigm with a fixed margin for the different users under each behavior, limiting the expressiveness of the model.

Motivated by these challenges, we propose a general and efficient heterogeneous learning framework, named \underline{C}riterion-guided \underline{H}eterogeneous \underline{C}ollaborative \underline{F}iltering (abbreviated as CHCF) for collaborative filtering model learning (e.g., MF \cite{hu2008collaborative}, GMF \cite{he2017neural} and LightGCN \cite{he2020lightgcn}) for multi-behavior scenarios. To cope with the heterogeneous implicit interaction, we first introduce upper and lower thresholds for different users and items to indicate different selection criteria, featured by the low-rank user heterogeneity matrix and item heterogeneity matrix. Moreover, we integrate criterion learning and user preference learning into a unified framework CHCF, which can be trained jointly for target interaction prediction. {To allow the diversity of preference scores for enhanced generalization, our heterogeneity-based criterion loss incorporates the hinge loss as well as heterogeneous selection criteria.} 
Importantly, we theoretically demonstrate that the loss of CHCF is an upper bound of collaborative metric learning \cite{hsieh2017collaborative} (CML) loss but in an unbiased non-sampling strategy, indicating that the optimization of Collaborative Metric Learning can be approximately achieved by the CHCF learning framework in a non-sampling form effectively. 
To evaluate the performance of our CHCF learning framework, we apply our CHCF to optimize three popular CF models (i.e., MF, GMF, LightGCN) with comprehensive experiments. And the results show our proposed CHCF significantly perform better than the state-of-the-art approaches on the multi-behavior recommendation task. Furthermore, we offer the training efficiency and empirical convergence of our framework, making it more practical and generalized in real-world heterogeneous scenarios.  
To summarize, the contributions of our work are as follows: 
\begin{itemize}
    \item  We propose CHCF, a non-sampling recommendation learning framework that models heterogeneous selection criteria and learns user preference simultaneously. We first introduce multiple user- and item-specific thresholds to feature different criteria for implicit heterogeneous data. Then we integrate the criterion learning and user preference learning into a unified framework, which makes full use of heterogeneous collaborative signals for user preference learning.
    \item We theoretically demonstrate that our CHCF framework can optimize the upper bound of CML in an unbiased non-sampling strategy, implying that the optimization of CML can be approximately achieved by our CHCF learning framework in an efficient non-sampling form. 
    \item To show the effectiveness of CHCF, we implement several variants of CHCF with three CF models to model the interaction. Comprehensive experiments on three real-world datasets demonstrate that our CHCF outperforms the state-of-the-art approaches in multi-behavior scenarios. 
\end{itemize}

%% file: related_work.tex
\section{Related Work}
\subsection{Collaborative Filtering with Implicit Data}
We begin with reviewing collaborative filtering methods in single-behavior scenarios~\cite{li2022causal,wang2020multi,wang2019dmran,qin2023learning}. Note that the majority of users could not rate items and thus it is sometimes hard to get explicit feedback from real-world data. Therefore, the number of implicit records including viewing, bookmarking, and purchasing items, largely surpasses the number of explicit feedback records including a like and a rating. To this end, it is crucial to design efficient recommendation algorithms that could tackle implicit feedback \cite{oard1998implicit,li2020iptv,li2020viewpoint}. Early efforts include ALS model \cite{hu2008collaborative} which factorizes the binary interaction matrix and assumes that users do not like unobserved items. SLIM~\cite{ning2011slim} learns the item-item similarity with additional constraints and some kernel methods~\cite{polato2018boolean} are developed to increase the expressive ability of RS. By presuming that users prefer the observed items to the unseen ones, the majority of recent methods consider interaction data as positive-only~\cite{rendle2009bpr,he2017neural}. 
Inspired by the process of deep learning, NCF~\cite{he2017neural} substitutes the inner product with a neural network capable of learning an arbitrary function from input, combining Generalized Matrix Factorization (GMF) and Multi-Layer Perceptron (MLP).
{On the basis of MF, ExpoMF~\cite{liang2016modeling} considers all unobserved interactions to be negative and leverages item popularity to weight them.}
Mult-VAE~\cite{liang2018variational} extends variational autoencoders to collaborative filtering for implicit feedback. With the development of graph learning methods \cite{kipf2016semi,velivckovic2017graph}, some graph neural network (GCN) based methods have been proposed to leverage the potential semantics in order to capture the high-order collaborative signals, such as Neural Graph Collaborative Filtering \cite{wang2019ngcf}, LightGCN \cite{he2020lightgcn} and etc. 
Although current collaborative filtering models have achieved great success in single-behavior scenarios, these CF models generally depend on negative sampling strategies for efficient model learning. Actually, our CHCF framework can optimize the above user preference CF models and applies them in heterogeneous scenarios.

\subsection{Multi-behavior Recommendation} 

The multi-behavior-based recommendation seeks to harness various types of interaction behavior to enhance the retrieval performance for the target behavior \cite{singh2008relational, loni2016bayesian}. The early representative approach CMF \cite{singh2008relational} proposes to share item embeddings across behaviors to concurrently factorize different interaction matrices and extends to exploiting different user behaviors in multi-behavior recommendation \cite{zhao2015improving}. Early works include eALS \cite{ ding2019sampler,gao2019learning} which exploits auxiliary interactions with an improved negative sampler. \R{Recently, multi-task learning frameworks have been widely adopted for multi-behavior recommendation. For example, ESM$^2$~\cite{wen2020entire} learns from decomposed pretexts to enhance the target prediction using the multi-task learning framework. DMFP~\cite{wang2019dmfp} utilizes multiple attention mechanisms to extract user preference from complementary perspectives.}
Efficient Heterogeneous Collaborative Filtering (EHCF) \cite{chen2020efficient} attempts to capture fine-grained user-item interaction and optimizes the network parameters efficiently from the whole heterogeneous data to further enhance the performance. These methods mostly regard auxiliary user-item behaviors as weaker interactions of user preference learning. Recent works usually tackle this task without any prior information about user preference strengths.
MBGCN \cite{jin2020multi} builds upon a message-passing architecture and designs an item-to-item embedding updating manner for exploring item-to-item similarity. {MATN~\cite{xia2020multiplex} encodes multiple relational structures via exploring the cross-behavior collaborative information as well as within-behavior collaborative signals. MB-GMN~\cite{xia2021graph} combines the exploration of multi-behavior relationships with the meta-learning paradigm.
GNMR~\cite{xia2021multi} explores the relationships among various types of behaviors using the message passing module.
HMG-CR~\cite{yang2021hyper} builds hyper meta-paths along with meta-graphs to explicitly capture the relationships among various behaviors.
KHGT~\cite{xia2021knowledge} constructs a knowledge-aware collaborative graph to achieve high-order relation exploration in the multi-behavior scenarios.} 
Graph Heterogeneous Collaborative Filtering (GHCF) \cite{chen2021graph} learns the representations of users and items, as well as relationships for multi-behavior forecasting and trains the model from the whole heterogeneous data, achieving state-of-the-art performance. \R{Recently, self-supervised techniques~\cite{chen2020simple,wang2021self} have also been extended to solve this problem~\cite{gu2022self,wu2022multi}, which identify the important behaviors as the pretext task.}
As pointed out in the introduction, the existing models do not well consider the key challenges, which are addressed by our neural network-based learning framework CHCF. Actually, our CHCF is a generalized multi-behavior recommendation learning framework, which is capable of optimizing any collaborative filtering model.

\subsection{Model Learning in Recommendation}
As we know, data matrices are highly sparse in many real-world applications. Two types of approaches have been widely utilized in recent studies to optimize the recommender system models as follows: (1) negative sampling approaches \cite{rendle2009bpr, gao2019learning} and (2) non-sampling approaches \cite{xiao2017learning,chen2019efficient,chen2020efficient,chen2020tois,chen2021graph}. 
The negative sampling approaches generate negative examples from  missing data. As a widely-used learning framework, BPR \cite{rendle2009bpr} selects negative examples out of missing entries at random and then maximizes the distance between the prediction of observed interaction and that of sampled negative samples. Unfortunately, existing sampling-based approaches typically converge slower and the design of the sampler is of great concern for the performance. Non-sampling approaches treat all missing interactions as negative. For instance, WMF \cite{hu2008collaborative} treats all missing data as negative examples with labels set to zero, but assigns them a low weight in the regression objective. Traditional non-sampling learning methods are mostly based on Alternating Least Squares (ALS) \cite{hu2008collaborative,he2016fast}. However, these methods cannot be applied to neural network-based methods. \R{Recently, ENMF \cite{chen2020tois} develops fast optimization methods to accumulate learning for neural recommendation models and various multi-behavior methods adopt the binary regression loss in the multi-task learning framework to achieve the state-of-the-art performance~\cite{chen2021graph,chen2020efficient}. However, these methods do not consider the various preference strengths of implicit data and heterogeneous selection criteria for different behaviors. To tackle this, we come up with a novel generalized  framework for learning multi-behavior recommender systems, which learns the selection criteria and user preference automatically. To the best of our knowledge, we are the first to learn CF models with the help of selection criterion learning in heterogeneous scenarios.}

%% file: method.tex
\section{Method}\label{sec:model}
\begin{table}[t]
\centering
\caption{{Summary of symbols and notation.\label{tab:sym}}}
	\begin{tabular}{c|c}
		\hline
		Symbol & Description   \\
		\hline
	    $U$ & set of users \\
	    $V$ & set of items \\
	    $d$ & embedding dimension \\ 
	    $K$ & number of behaviors \\
	    $\bm{R}^{(k)}$ & interaction matrices for behavior $k$ \\
	    $\bm{\hat{R}}$ & shared preference score matrices\\
	    $V_{u}^{(k)+}$& observed item set of user $u$ for behavior $k$\\
	    $V_{u}^{(k)-}$ & unobserved item set of user $u$ for behavior $k$\\
	     $S_{uv}^{(k)}$ & the upper threshold for user u and item v \\
	     $T_{uv}^{(k)}$  & the lower threshold for user u and item v \\
	     $\mathbf{H}$ & user heterogeneity matrix \\
	     $\mathbf{G}$ & item heterogeneity matrix \\
	    $\Theta_{CF}$ & parameters for collaborative filtering model \\
	     $\Theta_{HCL}$ & parameters for heterogeneous criterion learning model \\
		\hline
	
	\end{tabular}
\end{table}
In this part, we begin with formalizing the heterogeneous collaborative filtering problem. Second, we revise the existing non-sampling training strategies. Third, we elaborate on our CHCF to learn a specific CF model. Lastly, we provide some theoretical analysis of our proposed framework.

{Table \ref{tab:sym} summarizes the notations and key concepts.} Assuming a dataset consisting of $|U|$ users and $|V|$ items, the index $u$ and $v$ denote a user and an item respectively. Let $\left\{\bm{R}^{(1)}, \bm{R}^{(2)}, \ldots, \bm{R}^{(K)}\right\}$ represent the interaction matrices for $K$ types of behavior. Each entry in $\bm{R}^{(k)}$ is valued 1 or 0 according to the users' feedback of the $k$-th behavior. In formulation,
\begin{equation}
R_{uv}^{(k)}=\left\{
\begin{array}{ll}
1, & \text {if $u$ has interacted with $v$} \\
0, & \text {otherwise.}
\end{array}\right.
\end{equation}
In our problem, without loss of generality, the $K$-th behavior is regarded as the target behavior, which will be evaluated during inference. The task of heterogeneous collaborative filtering aims to predict whether the user $u$ would have an interaction with the item $v$ on the target behavior. All these items are ranked based on interaction scores, offering a recommendation list of items for every user. It is worth noting that there is no restriction on the temporal order or strength order of different behaviors. To put it another way, behavior $k-1$ does not have to occur before behavior $k$ and $R_{uv}^k = 1 $ does not indicate a stronger or weaker user preference compared with $R_{uv}^{k-1} = 1 $.  

\subsection{Preliminary: Non-Sampling Training Strategy  }
Most of the multi-behavior recommendations \cite{chen2020efficient, chen2021graph,chen2020tois} involve in two parts. The first part is the collaborative filtering model producing the likelihood for each type of behavior. The second part is the learning framework to optimize the collaborative filtering model. 
Specifically, after estimating the likelihood of $u$ conducting the k-th behavior on $v$ with $\hat{R}_{uv}^k$, recent state-of-the-art models generally utilize a non-sampling strategy by adopting a uniform weighted regression with squared loss \cite{hu2008collaborative}, and the loss of the $k$-th behavior for collaborative filtering (CF) learning is formulated as:
\begin{equation}
    \mathcal{L}_{CF}^{(k)}(\Theta_{CF}) = \sum_{u\in{U}}\bigg(\sum_{v\in{V}_{u}^{(k)+}} \Big(1-\hat{R}_{uv}^{(k)}\Big)^2 + w\sum_{v^{'}\in{V}_{u}^{(k)-}}{\hat{R}_{uv^{'}}^{(k)^2}}\bigg)
\label{equ:regress_loss}
\end{equation}
in which $\Theta_{CF}$ represents the parameters of the CF model, $V_{u}^{(k)+}$ ($V_{u}^{(k)-}$) is the set of positive (negative) items for user $u$ on the $k$-th behavior, $w$ denotes the weight of negative entry. 

However, for E-commerce implicit data, $R_{uv}^{(k)}=1$ implies that the $k$-th interaction exists between user $u$ and item $v$, but it does not imply that $u$ really likes $v$. Moreover, $R_{uv}^{(k)}=0$ does not necessarily imply that $u$ dislikes $v$, since perhaps $u$ has not been aware of $v$. In other words, $R_{uv}^{(k)}=R_{uv'}^{(k)}=1/0$ may imply different preference strength for items $v$ and $v'$. As a result, the above non-sampling learning strategy which learns parameters to regress the likelihood of each item to 1/0 is hard to distinguish true/false preferential items under $k$-th behavior from the implicit data.
What's more, for different types of positive (negative) interactions $R_{uv}^{(k)}$ and $R_{uv}^{(k')}$, it regresses their likelihood to the same value 1(0). In this way,
it also does not explore the heterogeneous information for different users and items, which greatly restricts the performance of multi-behavior recommendation.

\begin{figure}
\centering
\includegraphics[width=12cm,keepaspectratio=true]{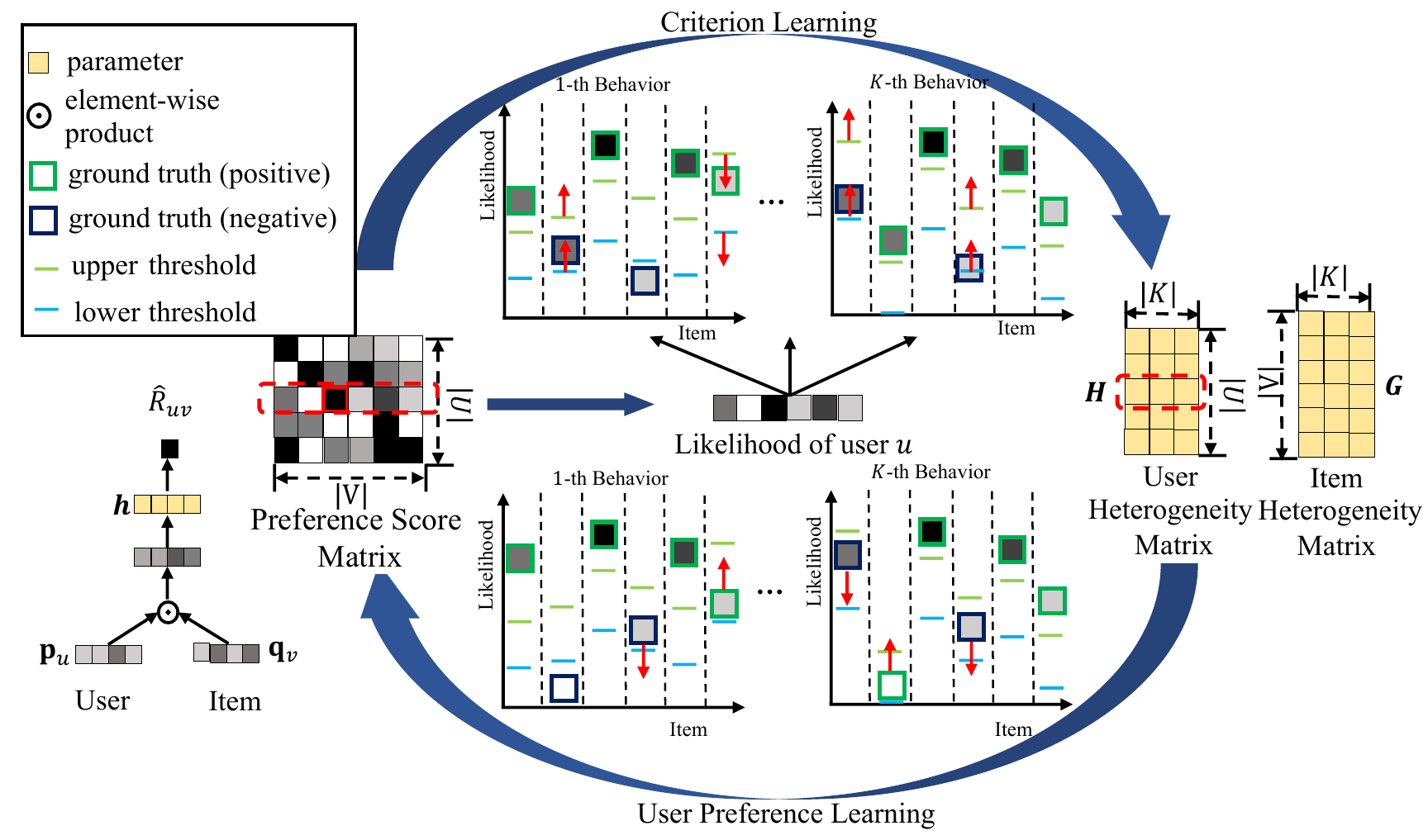}
\caption{The Architecture of GMF-CHCF. \R{We first utilize the collaborative model GMF (left) to generate a preference score for each user-item pair. The thresholds for each behavior are derived from the user and item heterogeneity matrices (right). Our \method{} compares the preference scores and thresholds based on observations in the loss objective.}  
The small red arrows represent the learning process. Given user u and item, GMF outputs the {preference score} $\hat{R}_{uv}$ on the left. The upper threshold $S_{uv}^{(k)}$ and lower threshold $T_{uv}^{(k)}$ for the $k$-th behavior are obtained by user heterogeneity matrix $\mathbf{H}$ and item heterogeneity matrix $\mathbf{G}$ on the right. The loss function punishes the situation when the {preference score} is less than upper threshold $S_{uv}^{(k)}$ for the positive interactions (i.e., $\hat{R}_{uv}<S_{uv}^{(k)}$) and the {preference score} is larger than lower threshold $T_{uv'}^{(k)}$ for the negative interactions (i.e., $\hat{R}_{uv'}>T_{uv'}^{(k)}$). We minimize the loss by jointly updating the parameters of GMF as well as $\mathbf{H}$ and $\mathbf{G}$.}

\label{fig:framwork}
\end{figure}
\subsection{Modeling Heterogeneous Selection Criteria}\label{ref:3-B}
To improve the non-sampling training strategy, we need to regress each positive entry to a specific target value. Nevertheless, from the perspective of ranking, if a positive interaction gets a higher {preference score} than the target value, it still suffers the regression loss, which is unreasonable. To design a more reasonable solution, hinge loss shall be considered to punish the wrong decision domains. The critical point of the hinge loss is called the decision bound, and 1 (positive) and 0 (negative) are the trivial bounds in Equation \ref{equ:regress_loss}. 
Moreover, we propose user- and item-specific decision bound, modeling the selection criteria for multiple types of behavior, which can reflect dependencies across behaviors. In other words, behavior dependencies are correlated with users' and items' specialties. For example, for each user $u$, if the relation of two types of behavior is close, one behavior's decision bound is similar to the other's. In reverse, for each user $u$, two similar obtained decision bounds mean that there are two similar behavior records, which further implies that the relation of two types of behavior is close. 

As such, we integrate user-and item-specific behavior dependency into threshold learning and propose a novel criterion loss. As shown in Fig. \ref{fig:framwork}, in our framework {all the behaviors share the same preference but have different selection criteria. By comparing the preference with different selection criteria, people behave differently with different probabilities}. To show how CHCF works, we optimize a specific preference learning model Generalized Matrix Factorization (GMF) \cite{he2017neural} with CHCF (denoted as GMF-CHCF). Firstly, users and items are mapped into the embedding space and GMF fuses user and item embeddings. In formulation, 
\begin{equation}
\phi\left(\mathbf{p}_{u}, \mathbf{q}_{v}\right)=\mathbf{p}_{u} \odot \mathbf{q}_{v},
\end{equation}
where $\mathbf{p}_{u} \in \mathbb{R}^d$ and $ \mathbf{q}_{v}\in \mathbb{R}^d$ are hidden embedding of $u$ and $v$, d is the embedding size, and $\odot$ represents the element-wise product of vectors. Lastly, GMF predicts users' preference through projecting the vector to the predicting layer. 
\begin{equation}
 \hat{R}_{u v} =\mathbf{h}^{T}\left(\mathbf{p}_{u} \odot \mathbf{q}_{v}\right) \\ =\sum_{i=1}^{d} h_{i} p_{u, i} q_{v, i} ,
\end{equation}
where $\mathbf{h} \in \mathbb{R}^d$ denotes the predict layer.

After getting the user-item preference scores through GMF, the heterogeneity-based criterion loss (HCL) of the $k$-th behavior is formulated as:
\begin{equation}
\begin{split}
\mathcal{L}_{HCL}^{(k)}(\Theta_{HCL}^{(k)}) = \sum_{u\in{U}}\bigg(\sum_{v\in{V}_{u}^{(k)+}} g\Big((S_{uv}^{(k)}-\hat{R}_{uv})_+\Big) \\+ w\sum_{v^{'}\in{V}_{u}^{(k)-}} g\Big((\hat{R}_{uv^{'}}-T_{uv}^{(k)})_+\Big)\bigg)
\label{equ:criterion_loss_user_item}
\end{split}   
\end{equation}
where $\Theta_{HCL}^{(k)}=\{S_{uv}^{(k)},T_{uv}^{(k)}\}_{u\in U}, v \in V$, $(\cdot)_+$ denotes $\max(\cdot,0)$. Here we use upper threshold $S_{uv}^{(k)}$ and lower threshold $T_{uv}^{(k)}$ as decision bound instead of 1 and 0 in Equation \ref{equ:regress_loss} to model heterogeneous selection criteria. Equation \ref{equ:criterion_loss_user_item} aims to punish the situation when the likelihood is less than upper threshold $S_{uv}^{(k)}$ for the positive interactions (i.e., $\hat{R}_{uv}<S_{uv}^{(k)}$) and the likelihood is larger than lower threshold $T_{uv'}^{(k)}$ for the negative interactions (i.e., $\hat{R}_{uv'}>T_{uv'}^{(k)}$). From the statistical view, for each user $u$, $[S_{uv}^{(k)}, +\infty)$ and $(-\infty, T_{uv'}^{(k)}]$ respectively denote the user-adaptive confidence interval of positive interaction $(u, v)$ and negative interaction $(u, v^{'})$ for the $k$-th behavior. Here $g(\cdot)$ is an arbitrary non-negative monotonous function (e.g., $g(x)=x$ or $x^{2}$) for decorating the distance between the scores and thresholds, which will be analyzed in the end of this section.

Note that $\mathbf{S}^{(k)}\in \mathbb{R}^{|U|\times |V|}$ is the upper threshold matrix and $\mathbf{T}^{(k)}\in \mathbb{R}^{|U|\times |V|}$ is the lower threshold matrix of the $k$-th behavior. For the user $u$, its selection criterion for the item $v$ depends on user idiosyncrasy and item properties. To model this essence and reduce threshold parameters $(|U|\times|V|\times{K}\times2)$, we introduce user heterogeneity matrix $\mathbf{H}\in\mathbb{R}^{|U|\times K}$ and item heterogeneity matrix $\mathbf{G}\in\mathbb{R}^{|V|\times K}$ to approximate threshold matrices: 
\begin{equation}
    S_{uv}^{(k)} = H_{uk}G_{vk},\quad T_{uv}^{(k)} = \alpha S_{uv}^{(k)}
\label{equ:decomposition}
\end{equation}
in which $0 \leq \alpha \leq 1$ is a hyper-parameter to control the ratio of two thresholds. Note that the user heterogeneity matrix $\mathbf{H}$ corresponds to user idiosyncrasy and the item heterogeneity matrix $\mathbf{G}$ corresponds to item properties as discussed in the introduction. 
Now, we revise Equation \ref{equ:criterion_loss_user_item} as follows:
\begin{equation}
\begin{split}
\mathcal{L}_{HCL}^{(k)} (\Theta_{HCL}^{(k)})= \sum_{u\in{U}}\bigg(\sum_{v\in{V}_{u}^{(k)+}} g\Big((H_{uk}G_{vk}-\hat{R}_{uv})_+\Big) \\
+ w\sum_{v^{'}\in{V}_{u}^{(k)-}} g\Big((\hat{R}_{uv^{'}}-\alpha H_{uk}G_{v^{'}k})_+\Big)\bigg)
\label{equ:criterion_loss}
\end{split}
\end{equation}
where $\Theta_{HCL}^{(k)}=\{H_{uk}, G_{vk}\}_{u\in U, v\in V}$. \RR{The primary objective of assuming that each upper threshold is proportional to the lower threshold is to decrease the number of parameters, which also contributes to effective loss optimization. In particular, negative interactions tend to reduce the thresholds, while positive interactions raise the upper thresholds. Our design can balance these tendencies during optimization, effectively learning the appropriate user and item heterogeneity matrices. Without this constraint, the upper bounds would continue to increase and lower bounds would continue to decrease to minimize the loss objective, ultimately leading to a trivial solution.}

\subsection{Criterion-guided Heterogeneous Collaborative Filtering}
Similarly, when the threshold matrices are fixed, $\mathcal{L}_{HCL}^{(k)}$ can be used to guide heterogeneous collaborative filtering model learning (HCF) for user preference learning. In formulation, 
\begin{equation}
\begin{split}
\mathcal{L}_{HCF}^{(k)} (\Theta_{CF})= \sum_{u\in{U}}\bigg(\sum_{v\in{V}_{u}^{(k)+}} g\Big((H_{uk}G_{vk}-\hat{R}_{uv}(\Theta_{CF}))_+\Big) \\
+ w\sum_{v^{'}\in{V}_{u}^{(k)-}} g\Big((\hat{R}_{uv^{'}}(\Theta_{CF})-\alpha H_{uk}G_{v^{'}k})_+\Big)\bigg)
\label{equ:CHCF_CF}
\end{split}
\end{equation}
where $\Theta_{CF}$ is the parameters of the optimized CF model (i.e., GMF).
Lastly, we integrate criterion learning and user preference learning into a unified framework. In formulation, 
\begin{equation}
\begin{split}
\mathcal{L}_{C H C F}^{(k)}(\Theta_{C F}, \Theta_{H C L}^{(k)})=  \sum_{u\in{U}}\bigg(\sum_{v\in{V}_{u}^{(k)+}} g\Big((H_{uk}G_{vk}-\hat{R}_{uv})_+\Big) \\
+ w\sum_{v^{'}\in{V}_{u}^{(k)-}} g\Big((\hat{R}_{uv^{'}}-\alpha H_{uk}G_{v^{'}k})_+\Big)\bigg)
\label{equ:CHCF}
\end{split}
\end{equation}
 The only difference between Equation \ref{equ:CHCF} and Equation \ref{equ:criterion_loss} is that $\Theta_{CF}$ is flexible or not.

According to the paradigm of Multi-Task Learning (MTL) \cite{argyriou2007multi} which optimizes the models of several related tasks in a joint manner, the ultimate objective function of the CHCF framework is written as:

\begin{equation}
\min_{\Theta_{HCF},\Theta_{HCL}}\mathcal{L}_{CHCF}(\Theta_{CF}, \Theta_{HCL}) = \sum_{k=1}^{K} \lambda_{k} \mathcal{L}_{CHCF}^{(k)}
\label{equ:loss_finall}
\end{equation}
where $\Theta_{HCL}=\{\mathbf{H}, \mathbf{G}\}$. Here the term $\lambda_k$ aims to adjust the effect of the $k$-th type of behavior on MTL, since the significance of each behavior type could be diverse for the problems in different domains. \R{We optimize $\Theta_{CF}$ and $\Theta_{HCL}$ in a parallel fashion. This joint optimization in related problems has been validated for convergence empirically~\cite{li2020symmetric, ma2020probabilistic,argo}. In formulation,
\begin{equation}\label{eq:11}
\Theta_{CF}^{t+1}=\Theta_{CF}^{t}-lr * \frac{\partial \mathcal{L}_{CHCF}(\Theta_{CF}^t, \Theta_{HCL}^t)}{\partial \Theta_{CF}^t}
\end{equation}
\begin{equation}\label{eq:12}
\Theta_{HCL}^{t+1}=\Theta_{HCL}^{t}-lr * \frac{\partial \mathcal{L}_{CHCF}(\Theta_{CF}^t, \Theta_{HCL}^t)}{\partial \Theta_{HCL}^{t}}
\end{equation}
where $t$ stands for the $t$-th iteration, and $lr$ denotes the learning rate. We also validate the empirical convergences of our optimization strategy in Sec. \ref{sec:eff}.}
Furthermore, we normalize all user embeddings and item embeddings to be constrained within the Euclidean ball, mitigating the challenge of `curse of dimensionality' \cite{li2020symmetric,park2018collaborative,tay2018latent}. \R{If there are definite cascading relationships between behaviors, we can add the constraint that the upper threshold for the latter behavior is always above that of the former behavior. In formulation, we have $S_{uv}^{(k+1)} \geq S_{uv}^{(k)}$ where behavior $k$ always happens before $k+1$. To achieve this, we enforce both $\bm{H}$ and $\bm{G}$ to satisfy $\bm{H}_{uv}^{(k+1)}\geq \bm{H}_{uv}^{(k)}$ and $\bm{G}_{uv}^{(k+1)}\geq \bm{G}_{uv}^{(k)}$ by adding a non-negative matrix on $\bm{H}_{uv}^{(k)}$ ($\bm{G}_{uv}^{(k)}$) to obtain $\bm{H}_{uv}^{(k+1)}$ ($\bm{H}_{uv}^{(k)}$).}

We additionally enforce that $\sum_{k=1}^K \lambda_k =1$ for convenience. {As for inference, we need to compare the user's preference score $\hat{R}_{uv}$ with both the personalized upper threshold and the lower threshold for the target behavior. Note that the upper threshold is proportional to the lower threshold and thus the results are the same for both thresholds. Hence, the final target prediction score is calculated by $\hat{R}_{uv}/(H_{uK}G_{vK})$ in our CHCF framework.}
The whole learning procedure is summarized in Algorithm \ref{alg1}, respectively.

\begin{algorithm}
\caption{CHCF's Training Algorithm}
\label{alg1}
\begin{algorithmic}[1]
\REQUIRE Users $U$ and items $V$;\\
\quad \ \ Interaction data $\left\{\mathbf{R}^{(1)}, \mathbf{R}^{(2)}, \ldots, \mathbf{R}^{(K)}\right\}$;\\
\quad \ \ threshold ratio $\alpha$;\\
\quad \ \ Weight of negative entry $w$;\\
\ENSURE Parameters $\Theta_{CF}$ and $\Theta_{HCL}=\{\mathbf{H},\mathbf{G}\}$;
\STATE Initialize parameters $\Theta_{CF}$ and $\Theta_{HCL}$;
\REPEAT
\STATE Sample $B$ users from $U$ to construct a mini-batch;
\STATE Generate the likelihood matrix ($B\times |V|$) by forward-propagating through the given collaborative filtering model;
\STATE Calculate the loss by Equation \ref{equ:loss_finall};
\STATE {Update parameters $\Theta_{CF}$ by Equation \ref{eq:11}};
\STATE {Update parameters $\Theta_{HCL}$ by Equation \ref{eq:12}};
\UNTIL convergence
\end{algorithmic}
\end{algorithm}

\subsection{Theoretical Analysis on CHCF}
In this section, we theoretically demonstrate that CHCF is closely related to CML \cite{hsieh2017collaborative}. Here we consider the single behavior situation.
We first review the framework of CML, whose loss with adaptive margin \cite{li2020symmetric} is formulated as:
\begin{equation}
\mathcal{L}_{CML}=\sum_{u \in U}\sum_{v \in V_{u}^{+}}\sum_{v^{'} \in V_{u}^{-}} g\bigg(\Big({d(u, v)}^2-{d(u, v^{'})}^2+m_u\Big)_{+}\bigg)
\label{equ:cml_1}
\end{equation}
Obviously, $1-\hat{R}_{uv}$ is a reasonable description of distance metric ${d(u,v)}^2$. We first define a decorated function class.
\begin{definition}
$g(x)$ is called a low-order decorated function if $g(x)$ satisfies 
\begin{itemize}
    \item $g(0)=0$
    \item $g(x)$ is monotonically increasing.
    \item There exists a constant $M$, s.t. $\sup_{x>0} g(2x)< Mg(x)$
\end{itemize}
\end{definition}
Note that polynomials with $g(0)=0$ ($g(x)= x,x^2,x^3,\dots$), $g(x)= log(x+1)$ both belong to this category while $g(x)= e^x-1$ breaks the last condition. Moreover, this class is closed under addition and multiplication. We have the following theorem.
\begin{theorem}
If $g(x)$ is a low-order decorated function, CHCF is an upper bound of CML with whole-data strategies multiplied by a constant $C$. In formation,
\begin{equation}
\mathcal{L}_{CML} \leq C \cdot \mathcal{L}_{CHCF}
\end{equation}

\end{theorem}
\begin{proof}
If we consider $m_u=S_{uv}-T_{uv}$, Equation \ref{equ:cml_1} is revised as 
\begin{equation}
\mathcal{L}_{CML}=\sum_{u \in U}\sum_{v \in V_{u}^{+}}\sum_{v^{'} \in V_{u}^{-}} g\Big((S_{uv}-\hat{R}_{uv}+\hat{R}_{uv^{'}}-T_{uv^{'}})_{+}\Big)
\label{equ:cml_2}
\end{equation}
Without loss of generality, we assume $a \leq b$, then 
\begin{equation*}
g((a+b)_+) \leq  g(2b_+) \leq  M g(b_+) \leq  M(g(a+) + g(b+)) 
\end{equation*}
Therefore, if we set $ w = |V_{u}^{+}|/|V_{u}^{-}|$, the following
\begin{equation}
\begin{split}
\mathcal{L}_{CML} \leq & M  \sum_{u \in U}\sum_{v \in V_{u}^{+}}\sum_{v^{'} \in V_{u}^{-}}g\Big((S_{uv}-\hat{R}_{uv})_{+}\Big)\\ &+g\Big((\hat{R}_{uv^{'}}-T_{uv^{'}})_{+}\Big) \\
\leq & M |V_{u}^{-}|\sum_{u\in{U}}\bigg(\sum_{v\in{V}_{u}^{(k)+}} g\Big((S_{uv}-\hat{R}_{uv})_+\Big) \\&+ \frac{|V_{u}^{+}|}{|V_{u}^{-}|}\sum_{v^{'}\in{V}_{u}^{(k)-}} g\Big((\hat{R}_{uv^{'}}-T_{uv^{'}})_+\Big)\bigg)\\
= & C \mathcal{L}_{CHCF} 
\end{split}
\label{equ:final}
\end{equation}
holds. 
\end{proof}

From Theorem 1, when taking all triplets into consideration, the upper bound of the CML loss is proportional to the CHCF loss. As a result, the minimization of CML can be approximately achieved by the optimization of CHCF \cite{yu2020sampler}.  
 
\R{\noindent\textbf{Efficiency of CHCF.} We take $g(x)=x$ as an example and other decorated functions lead to similar results. In Equation \ref{equ:CHCF}, updating all the users for each behavior takes $O(|U||V|(d+1))$ time.
For CML loss, it takes $O(2\sum_u|V_u^+||V_u^-|d)$ time. Since $2|V_u^+||V_u^-| >> |V|)$, the computational complexity of our approach is decreased by several magnitudes. Actually, in practice, CML always accompanies a biased negative sampling strategy, which hinders from the optimal ranking performance even though extensive updating steps have been conducted \cite{xin2018batch}. Besides, CHCF shares the same preference scores for different behaviors, which results that updating the whole model through MTL for all the $K$ types of behavior takes $O(2\sum_u|V_u^+||V_u^-|(d+K))$ time. Since $d>>K$ in most practical cases, our model is more efficient than most of the existing multi-behavior state-of-the-arts (e.g., NMTR~\cite{gao2019neural}, GHCF~\cite{chen2021graph}), whose computational complexity is proportional to the number of behaviors $K$. 
Above all, our CHCF optimizes the model with an efficient whole-data strategy. 
}

\noindent\textbf{Choices of $g(x)$.} We suggest using a strictly convex decorated function, which has a monotonically increasing gradient. In formulation, \begin{equation}
    x_1 > x_2 \geq 0 \Rightarrow g'(x_1) > g'(x_2)
\end{equation}
This implies that if the interaction score breaks the criterion seriously (i.e., $S_{uv} >> \hat{R}_{uv}$ or $\hat{R}_{uv'} >> T_{uv'}$ in Equation \ref{equ:criterion_loss_user_item} ), we give this term a larger gradient. In this way, CHCF gives more emphasis to 'bad' interactions, which helps the optimization process and further increases the performance. Similar results can be obtained by CML. In our implementation, we use $g(x) = x^2 $ as default for fair comparison with regression loss in \cite{chen2020efficient,chen2021graph}. 

\subsection{Implementation}
 To optimize the objective function, we use mini-batch Adagrad \cite{duchi2011adaptive} as the optimizer. Note that in the methods, we first choose GMF-CHCF as an example to illustrate our CHCF learning framework. Actually, our framework can optimize arbitrary user preference CF models. We also optimize two extra models Matrix Factorization and LightGCN in the experiments. MF is the simplest but effective CF model, which predicts user preference by the dot product of latent embeddings, i.e., $\hat{R}_{uv} = \mathbf{p}_u \odot \mathbf{q}_v$. MF-CHCF can show the ability of CHCF for optimizing simple models. LightGCN is a neural model which achieves state-of-the-art performance using GCN. LightGCN-CHCF can show the performance and empirical convergence of CHCF for optimizing complex models. In practice, we can select different CF models learned by our CHCF framework, and choose the best based on validation. 

%% file: experiment.tex
\section{Experiments}\label{sec:exp}
In this part, we conduct extensive experiments on three datasets to validate the effectiveness of our CHCF. Here, we will answer seven research questions as follows:
\begin{itemize}
\item  \textbf{RQ1}: In comparison to state-of-the-art method, does our CHCF obtain better performance?
\item  \textbf{RQ2}: How is the scalability of the CHCF framework?
\item  \textbf{RQ3}: How is the effectiveness of our proposed techniques (hinge loss design, user heterogeneity and item heterogeneity modeling) and different types of auxiliary behaviors?
\item  \textbf{RQ4}: Can our model ease the sparse data issue, that users have few records on the target behaviors?
\item  \textbf{RQ5}: How do the model hyper-parameters (bound ratio $\alpha$, negative entry weight $w$ and loss coefficient of behavior $\lambda_k$) affect the final performance of CHCF?
\item  {\textbf{RQ6}: How does our model perform in the cold-start setting?}
\item  \textbf{RQ7}: What behavior-relation patterns does our model capture after training?
\end{itemize}

\subsection{Experimental Settings}

\subsubsection{Datasets}
We perform extensive experiments on popular real-world benchmark which contain multiple types of behaviors. The details of these datasets are presented as below:
\begin{itemize}
    \item \textbf{BeiBei}\footnote{\url{https://github.com/chenchongthu/EHCF}}. It is a dataset collected from Beibei, which is a large website for purchasing infant products in China. It contains 21716 users and 7977 items with different types of behaviors, i.e., purchase, cart and view.
    \item \textbf{Taobao}\footnote{\url{https://tianchi.aliyun.com/dataset/dataDetail?dataId=649}}. It is a dataset collected from Taobao, which is the largest Chinese e-commerce platform. It contains 48749 users and 39493 items encompassing different types of behaviors.
    \item {\textbf{Tmall}\footnote{\url{https://tianchi.aliyun.com/dataset/dataDetail?dataId=47}}. It is also a dataset from the Chinese e-commerce platform. It includes four categories of behavior, i.e., page view, cart, mark-as-favorite and purchase.}
\end{itemize}
\begin{table}[t]
\caption{{Statistical details of the evaluation datasets.}}
\setlength\tabcolsep{2pt}
\begin{tabular}{lrrrr}
\textbf{Dataset} & \textbf{\#User} & \textbf{\#Item} & \textbf{\#Interaction} & \textbf{\#Interactive Behavior Type} \\
\hline
\emph{\textbf{Beibei}} & 21,716 & 7,977 &$3.4\times 10^6$  & $\{$Page View, Cart, Purchase$\}$ \\ 
\emph{\textbf{Taobao}} & 48,749 & 39,493 &$2.0\times 10^6$ &  $\{$Page View, Cart, Purchase$\}$ \\ 
\emph{\textbf{Tmall}} & 9,368 & 302,722 & $1.6\times 10^6$ & $\{$Page View, Favorite, Cart, Purchase$\}$ \\ 
\hline
\end{tabular}
\label{tab1}
\end{table}

Following \cite{chen2020efficient,chen2021graph,jin2020multi,gao2019neural}, we combine the duplicated user-item interactions and further filter out users and subsequently, users and items with fewer than five purchase interactions are excluded. The final purchase records of users are chosen as test data, while the second-last records are allocated for validation. The remaining ones are utilized to form the training set. We summarize the statistical details of three datasets in Table \ref{tab1}.

\subsubsection{Baselines}
We compare our CHCF with a wide range of state-of-the-art approaches. They can be categorized into two types, i.e., one-behavior approaches which merely leverage target behavior records, and multi-behavior approaches which leverage all types of behaviors. As for these baseline models, we consult their corresponding papers and follow their procedures to tune the parameters.
One-behavior models include: 
\begin{itemize}
	\item \textbf{BPR} \cite{rendle2009bpr}, it trains MF using the Bayesian Personalized Ranking (BPR) loss function. 
	\item \textbf{ExpoMF} \cite{liang2016modeling}, it is a whole-data-based MF approach for item recommendation, which considers all missing data as negative with popularity-based weights.
	\item \textbf{NCF} \cite{he2017neural}, it is a popular deep learning approach that implements MF using a Multi-Layer Perceptron (MLP).
	\item \textbf{ENMF} \cite{chen2020efficient}, it is a non-sampling learning recommendation framework which optimizes GMF model with an efficient non-sampling strategy for Top-N recommendation.
	\item \textbf{LightGCN} \cite{he2020lightgcn}, it is the state-of-the-art GNN-based recommender system approach which modifies the design of GNN from NGCF \cite{wang2019ngcf} for a more concise and appropriate architecture.
\end{itemize}
Multi-behavior models include:
\begin{itemize}
\item \textbf{CMF} \cite{zhao2015improving}, it decomposes the interaction matrices of different behavior types at the same time to combine the multiple factorization procedures via sharing the representations of common entities.
\item \textbf{MC-BPR} \cite{loni2016bayesian}, it extends the origin BPR approach to fit heterogeneous scenarios.
\item \textbf{NMTR} \cite{gao2019neural}, it considers the behavior correlations in a cascaded manner and captures multi-type interactions with a multi-task learning framework.
\item \textbf{EHCF} \cite{chen2020efficient}, it explores fine-grained relationships and efficiently optimizes the network with non-sampling strategies. EHCF is a heterogeneous version of ENMF, which is also based on GMF. 
\item \textbf{MBGCN} \cite{jin2020multi}, it considers modeling the strengths of different behaviors and exploits the data with a graph convolutional network.
\item \textbf{MATN} \cite{xia2020multiplex}, it effectively explores multi-behavior relational structures via maintaining the collaborative signals across different types of behaviors.
\item \textbf{GHCF} \cite{chen2021graph}, it takes the advantages of GNNs to learn the embedding of users, items and relations for comprehensive multi-behaviour recommendation.
\item \textbf{GNMR}~\cite{xia2021multi}, it explores the relationships among various types of behaviors using the message passing module.
\item \textbf{HMG-CR}~\cite{yang2021hyper}, it builds hyper meta-paths along with meta-graphs to explicitly capture the relationships among various behaviors.
\item \textbf{MB-GMN}~\cite{xia2021graph}, it combines the exploration of multi-behavior relationships with the meta-learning paradigm.
\end{itemize}

\begin{table*}[!ht]
\centering
\caption{\R{Performance of different models and the improvement rate compared with the best baseline on Beibei and Taobao.} }\label{result}
\begin{tabular*}{0.99\textwidth}{@{\extracolsep{\fill}}lcccccc@{}}
\toprule
\multirow{2}{*}{\textbf{Beibei}} & \multicolumn{3}{c}{HR} & \multicolumn{3}{c}{NDCG} \\
\cline{2-4}\cline{5-7} 
 & \textbf{10} & \textbf{50} & \textbf{100}  & \textbf{10}& \textbf{50 } & \textbf{100}  \\
 \hline
\textbf {BPR} &0.0437 & 0.1246 & 0.2192  & 0.0213 & 0.0407 & 0.0539  \\ 
\textbf{ExpoMF} & 0.0452 & 0.1465 & 0.2246 & 0.0227 & 0.0426 & 0.0553 \\
 \textbf{NCF} & 0.0441 & 0.1562 & 0.2343 & 0.0225 & 0.0445 & 0.0584 \\ 
 \textbf{ENMF} & 0.0464 & 0.1637 & 0.2586 & 0.0247 & 0.0484&  0.0639 \\
 \textbf{LightGCN} & 0.0451 & 0.1613& 0.2495 & 0.0232 & 0.0466 & 0.0611 \\
\midrule
\textbf{CMF}  & 0.0420 & 0.1582 & 0.2843 & 0.0251 & 0.0462 & 0.0661 \\
\textbf{MC-BPR} & 0.0504 & 0.1143 & 0.2755 & 0.0540 & 0.0503 & 0.0653 \\
\textbf{NMTR} & 0.0524 & 0.2047 & 0.3189  & 0.0285 & 0.0609 & 0.0764 \\
\textbf{MBGCN} &  {0.1564} &  {0.3434} & 0.4262 &  {0.0828} &  {0.1282} &  {0.1384} \\
\textbf{EHCF} & 0.1523 & 0.3316 &  {0.4312}  & 0.0817 & 0.1213 & 0.1374  \\
\textbf{GHCF} & 0.1922 & 0.3794 & 0.4711 & 0.1012 & 0.1426 & 0.1575  \\
\hline
\textbf{MF-CHCF} & \textbf{0.2053}& \textbf{0.4300}  & \textbf{0.5421} & \textbf{0.1031}  & \textbf{0.1568} &\textbf{0.1746} \\
Improvement & 6.82\% & 13.34\% & 15.07\% & 1.88\% & 9.96\% &10.86\% \\
\textbf{GMF-CHCF} & \textbf{0.2308}& \textbf{0.4807}  & \textbf{0.5731} & \textbf{0.1189}  & \textbf{0.1740} &\textbf{0.1890}\\
Improvement & 20.08\% & 26.70\% & 21.65\% & 17.49\% & 22.02\% &20.00\% \\
\textbf{LightGCN-CHCF} & \textbf{0.2436}& \textbf{0.4324}  & \textbf{0.5490}  & \textbf{0.1340}  & \textbf{0.1755} &\textbf{0.1896} \\
Improvement & 26.74\% & 13.97\% & 16.54\% & 32.41\% & 23.07\% &20.38\% \\
\bottomrule
\multirow{2}{*}{\textbf{Taobao}} & \multicolumn{3}{c}{HR} & \multicolumn{3}{c}{NDCG} \\
\cline{2-4}\cline{5-7} 
 & \textbf{10} & \textbf{50} & \textbf{100}  & \textbf{10}& \textbf{50 } & \textbf{100}  \\
 \hline
 \textbf{BPR} & 0.0376 & 0.0708 & 0.0871 & 0.0227 & 0.0269 & 0.0305  \\ 
\textbf{ExpoMF} & 0.0386 & 0.0013 & 0.0911 & 0.0238 & 0.0220 & 0.0302 \\
\textbf{NCF} & 0.0391 & 0.0728 & 0.0897  & 0.0233 & 0.0281 & 0.0321 \\ 
 \textbf{ENMF} & 0.0398 & 0.0743 & 0.0936 & 0.0244 &0.0298 & 0.0339\\
 \textbf{LightGCN} & 0.0415 & 0.0814 & 0.1025  & 0.0237 & 0.0325 & 0.0359 \\
\hline
\textbf{CMF}  & 0.0483 & 0.0774 & 0.1185  & 0.0252 & 0.0293 & 0.0357  \\
\textbf{MCBPR} & 0.0547 & 0.0091 & 0.1264 & 0.0223 & 0.0227 & 0.0366  \\
\textbf{NMTR} & 0.0585 & 0.0942 & 0.1368 & 0.0278 & 0.0334 & 0.0394  \\
\textbf{MBGCN} & 0.0701 & 0.1522  & 0.2169   & 0.0390 & 0.0571 & 0.0653  \\
\textbf{EHCF} &  {0.0717} &  {0.1618} &  {0.2211}  &  {0.0403} &  {0.0594} &  {0.0690}  \\
\textbf{GHCF} & 0.0807 & 0.1892 & 0.2599 & 0.0442 & 0.0678 & 0.0792 \\
\hline
\textbf{MF-CHCF} & \textbf{0.1465}& \textbf{0.2416}  & \textbf{0.2829}  & \textbf{0.0891}  & \textbf{0.1102} &\textbf{0.1169} \\
Improvement & 81.54\% & 27.70\% & 8.85\% & 101.58\% & 62.54\% &47.60\% \\
\textbf{GMF-CHCF} & \textbf{0.1514} & \textbf{0.2460}  & \textbf{0.2843}  & \textbf{0.0916}  & \textbf{0.1117} &\textbf{0.1189}  \\
Improvement & 87.61\% & 30.02\% & 9.39\% & 107.23\% & 64.75\% &50.13\% \\
\textbf{LightGCN-CHCF} & \textbf{0.1571}& \textbf{0.2480}  & \textbf{0.2882}  & \textbf{0.0924}  & \textbf{0.1127} &\textbf{0.1192} \\
Improvement & 94.67\% & 31.08\% & 10.89\% & 109.05\% & 66.22\% &50.51\% \\

\bottomrule
\end{tabular*}
\label{tab:result}
\end{table*}

\subsubsection{Evaluation Metrics}
After training a model, we evaluate the retrieval performance on the target behavior~\cite{gao2019neural,chen2020efficient,chen2021graph} by generating the personalized ranking list. Here we rank all items that are not interacted with every given user in the dataset. We utilize the leave-one-out evaluation strategy as in \cite{gao2019neural,chen2020efficient,chen2021graph} and two widely-used metrics, HR@$N$\cite{karypis2001evaluation} and NDCG@$N$ \cite{jarvelin2000ir} are adopted to investigate the ranking performance. 
For the first two datasets, all the unobserved items in the dataset are ranked by our evaluation protocol as in \cite{chen2020efficient,chen2021graph}.
Under this protocol, the results can be more persuasive compared with merely ranking a random subset of negative items \cite{gao2019neural,chen2021graph, krichene2020sampled}. 
For the last large-scale dataset, \R{each positive item is paired with the same $99$ non-interactive items for each user randomly selected by existing works~\cite{xia2021multi,xia2021graph,yang2021hyper} for fair comparison.}
Besides, we vary the $N$ value when evaluating the retrieval performance. For every approach, the model is initialized randomly and executed five times and then the average performance is reported. \RR{Note that $N$ is from references~\cite{chen2020efficient,chen2021graph,yang2021hyper} to enable clear comparisons within reasonable ranges. } 

\subsubsection{Parameters Settings}
Our CHCF\footnote{Our source code is available at \url{http://github.com/luoxiao12/CHCF}.} framework is implemented with Pytorch. All hyper-parameters are tuned in the selected validation set. The parameters of baselines are initialized following their relevant papers, and are then elaborately tuned to obtain optimal performances. By default, our CHCF is configured with the embedding size $64$ and $16$ on the first two datasets and the last dataset, respectively advised by \cite{chen2021graph,yang2021hyper}.
The batch size is searched in $[128, 256, 512, 1024]$. The learning rate is tuned among $[0.005, 0.01, 0.02, 0.05]$. We set the negative sampling ratio as $4$ as in \cite{chen2020efficient,chen2021graph,gao2019learning} for sampling-based methods. For non-sampling methods, we set the weight of missing interactions as $w$, with the candidate set $[0.001, 0.01, 0.1, 0.2, 0.5, 1]$. For the coefficient parameter $\lambda_k$ in the final loss function, we tune them in $[0,1/6,2/6,3/6,4/6,5/6,1]$. As the only newly-introduced hyper-parameter, threshold ratio $\alpha$ is tuned among $[0,0.1,0.3,0.5,0.7,0.9,1]$.  

We set the batch size to $512$ and the learning rate to $0.05$ after the tuning procedure. $w$ is set to $0.1$ for Beibei and $0.01$ for the other two datasets. For Beibei and Taobao, the coefficient parameter $\lambda_k$ is set to $1/6$, $4/6$ and $1/6$ for view, cart and purchase behaviors, respectively. For Tmall, the coefficient parameter $\lambda_k$ is set to $1/6$, $1/6$, $1/2$ and $1/6$ for view, favorite, cart and purchase behaviors, respectively. The introduced hyper-parameter, threshold ratio $\alpha$ is set to $0.5$ by validation. The effects of these hyper-parameters will be further investigated in Section \ref{sec:hyper}.

\subsection{Performance Comparison (RQ1)}

We first compare our framework with the state-of-the-art methods by investigating the Top-N performance with N set to $[10, 50, 100]$ following \cite{chen2021graph}. We show the results of our CHCF based on three single-label recommender methods: MF (MF-CHCF), GMF (GMF-CHCF), and LightGCN (LightGCN-CHCF). The compared results of different methods on Beibei and Taobao are shown in Table \ref{tab:result} and we can make the following observation from the results:
1) \textit{Multi-behavior models work well than one-behavior models mostly.} According to the comparison of single-behavior approaches and multi-behavior approaches, we can find that introducing multi-behavior information to predictions can improve performance. In our baselines, the best multi-behavior approach, i.e., GHCF, is able to perform better than the best one-behavior approach, i.e., NCF, on Beibei by 82.1\% on HR@100 and 146.5\% on NDCG@100, while 153.6\% and 110.6\% on Taobao, validating the effectiveness of integrating multi-behavior signals into the approach.
2) \textit{Models adpoting whole-data-based optimization strategies (EHCF, GHCF and CHCF) outperform other sampling-based models (NMTR and MBGCN) generally.} {
For instance, NMTR and GMF-CHCF both optimize the GMF framework and NMTR achieves much worse performance using the negative sampling strategy. This clarifies the advantages of whole-data-based learning strategies since the sampling-based methods could fail to capture enough collaborative filtering signals and bring in biased results in multi-behavior scenarios.}
3) \textit{Our CHCF significantly performs better than other competitive approaches in terms of all evaluated metrics.}
For example, compared with GHCF, a recently developed and very powerful GCN model, our GMF-CHCF shows an improvement of 16.8\% and 17.0\% for HR@100 and NDCG@100 on the Beibei dataset and 8.85\% and 50.1\% on Taobao dataset respectively, which justifies the superiority of our model. Note that our GMF-CHCF utilizes the shallow and effective GMF model with much fewer parameters. Furthermore, compared with EHCF which has the same network structure but uses a different optimization framework, our GMF-CHCF outperforms it by 32.9\% and 28.6\% for HR@100 on the two datasets, indicating the potential of enhancing classic shallow approaches with a better optimization framework. 
4) \textit{The performance of CHCF is related to the basic CF models.} In comparison with three CF models optimized by the CHCF framework, we find that the performance in multi-behavior scenarios seems to be influenced by the choice of different basic CF models. One piece of evidence is that MF-CHCF achieve the  worst performance among the three CHCF models, which could result from the poor capacity of MF in learning collaborative filtering signals. Among three CHCF methods, LightGCN-CHCF achieves the best performance in most cases owing to both LightGCN's ability to capture collaborative filtering signals and CHCF's efficient learning framework. 

\begin{table*}[!ht]
\setlength{\abovecaptionskip}{0.2cm}
\centering
\caption{\R{Performance of different models and the improvement rate compared with the best baseline on Tmall.} }\label{result:2}
\begin{tabular}{lcccc}
\toprule
\textbf{Method}
 & \textbf{HR@5} & \textbf{HR@10} & \textbf{NDCG@5} & \textbf{NDCG@10}\\
 \hline
 
\textbf{NMTR} & 0.2780 & 0.4230 & 0.1798 & 0.2265  \\
\textbf{EHCF} & 0.3327 & 0.4209 & 0.2447 & 0.2732  \\
\textbf{GHCF} & 0.3588 & 0.4562 & 0.2889 & 0.3119  \\
\textbf{MB-GMN} & 0.3384 & 0.4508 & 0.2399 & 0.2824  \\
\textbf{HMG-CR} & 0.3547 & 0.4313 & 0.2642 & 0.2891 \\
\textbf{GNMR} & 0.3654 & 0.4415 & 0.2686 & 0.3029  \\
\hline
\textbf{MF-CHCF} & \textbf{0.3776}& \textbf{0.4629}  & \textbf{0.2986} & \textbf{0.3166} \\
Improvement & 3.34\% & 1.47\% & 3.36\% & 1.51\% \\
\textbf{GMF-CHCF} & \textbf{0.3818} & \textbf{0.4690} & \textbf{0.3039} & \textbf{0.3288} \\
Improvement & 4.49\% & 2.81\% & 5.19\% & 5.42\% \\
\textbf{LightGCN-CHCF} & \textbf{0.3996}& \textbf{0.5195}  & \textbf{0.2914} & \textbf{0.3301}\\
Improvement & 9.36\% & 13.88\% & 0.87\% & 5.83\% \\
\bottomrule
\end{tabular}

\label{tab:result3}
\end{table*}

\begin{figure*}[t]
    \centering
    \includegraphics[width=0.9\textwidth]{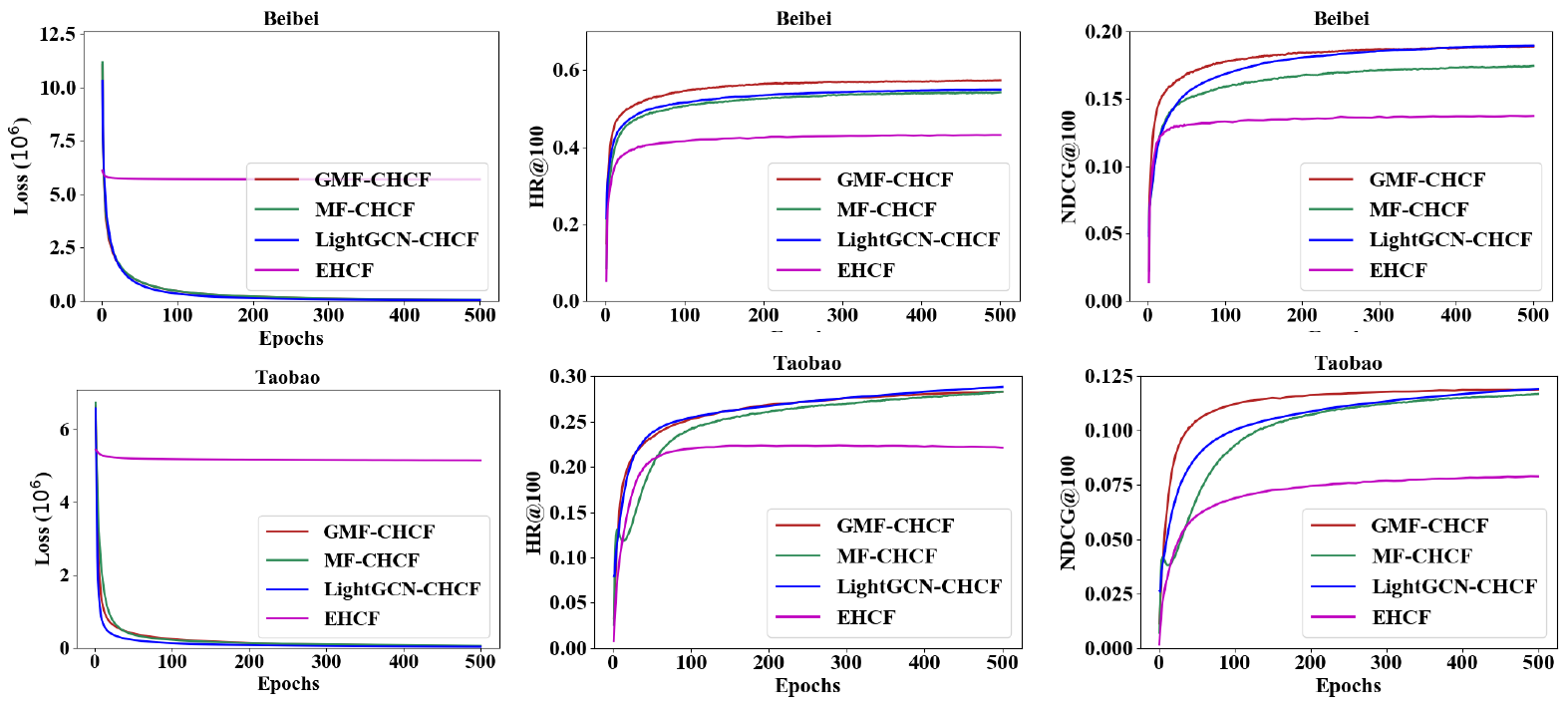}
    \caption{Training loss, HR@100 and NDCG@100 of EHCF and our three CHCF methods in each iteration on Beibei and Taobao.} 
    \label{fig:loss}
\end{figure*}

Moreover, we conduct experiments on Tmall by comparing our CHCF with recent state-of-the-art methods including NMTR, EHCF, GHCF, MB-GMN, HMG-CR and GNMR. We adopt the metrics as in \cite{yang2021hyper} and the compared results are shown in Table \ref{tab:result3}. From the results, we can find that our CHCF still outperforms recent state-of-the-art methods when it comes to a huge number of items, which indicates the superiority of our CHCF.

\subsection{Efficiency Analyses of CHCF (RQ2)}\label{sec:eff}

In this part, we evaluate the computing time of the proposed CHCF in comparison to three state-of-the-art MTL framework-based methods, NMTR, EHCF and GHCF. All these approaches are trained on a single NVIDIA GeForce GTX TITAN X GPU for a fair comparison of the efficiency.

\begin{table}[t]
\caption{Computational time cost investigation (second/minute [s/m]). Here "E" and "T" represent the training time for each epoch and the total training time, respectively. \R{We can observe the competitive model scalability of our CHCF.}}
\begin{tabular}{lrrrr}
\hline
\textbf{Models} & \textbf{Beibei/E} & \textbf{Beibei/T} & \textbf{Taobao/E} & \textbf{Taobao/T } \\
\hline
NMTR & 165s & 550m & 180s & 600m \\ 
EHCF & 7s & 24m & 16s & 54m \\ 
GHCF & 13($\pm$1)s & 45m & 34($\pm$1)s & 115m\\
GMF-CHCF & 7($\pm$1)s & 24m & 20($\pm$1)s & 70m \\
MF-CHCF & 7($\pm$1)s & 23m & 20($\pm$1)s & 68m \\
LightGCN-CHCF & 8($\pm$1)s & 28m & 22($\pm$1)s & 75m \\
\hline
\end{tabular}
\label{tab:eff}
\end{table}

 In Table \ref{tab:eff}, we show the running time of every epoch and total training during the training phase of each compared approach. It is observed that CHCF could achieve better performance than NMTR. 
 The reason is that our model integrates multiple types of behaviors into low-rank matrix learning, and shares the same interaction scores for different behaviors, which greatly increases the computational efficiency. GHCF improves the performance with its complex GCN architecture, with the cost of low efficiency. Although we lose in the cases when compared with efficient EHCF, we think that the speed-up of EHCF is based on its simple regression loss, which is proved hard to distinguish different preference strengths in implicit data. Above all, our CHCF shows a competitive model scalability because of the comparable computational complexity and great performance.

 Furthermore, we study the learning process of the efficient method EHCF and our three CHCF methods. Fig. \ref{fig:loss} shows the state of every approach with embedding size $64$ on two datasets. Note that besides training loss, we merely report the performance in terms of HR@100 and NDCG@100. For other metrics, we have a similar observation. From the figure, although the loss of EHCF converges faster, it cannot well fit the multi-behavior scenario and leads to low performance. Compared with EHCF, the initial training losses of CHCF models are huge but they decrease quickly and converge to small values. Moreover, our framework CHCF consistently achieves better performance in terms of HR and NDCG, which further demonstrates the superiority of our framework. Lastly, we find that complex models (e.g., LightGCN) do not influence the convergence of CHCF. In other words, our CHCF framework is a generalized learning framework, which promises empirical convergence.

\subsection{Ablation Study (RQ3)}
Since three CHCF models have similar performance generally, we utilize GMF-CHCF as an example in the following analysis\footnote{In the rest of this paper, we omit “GMF-” for concise representation except as otherwise noted.}. For other models, we have a similar observation.

\subsubsection{Comparisons of Different Decorated Functions}

\begin{figure}
    \centering
    \includegraphics[width=0.7\textwidth]{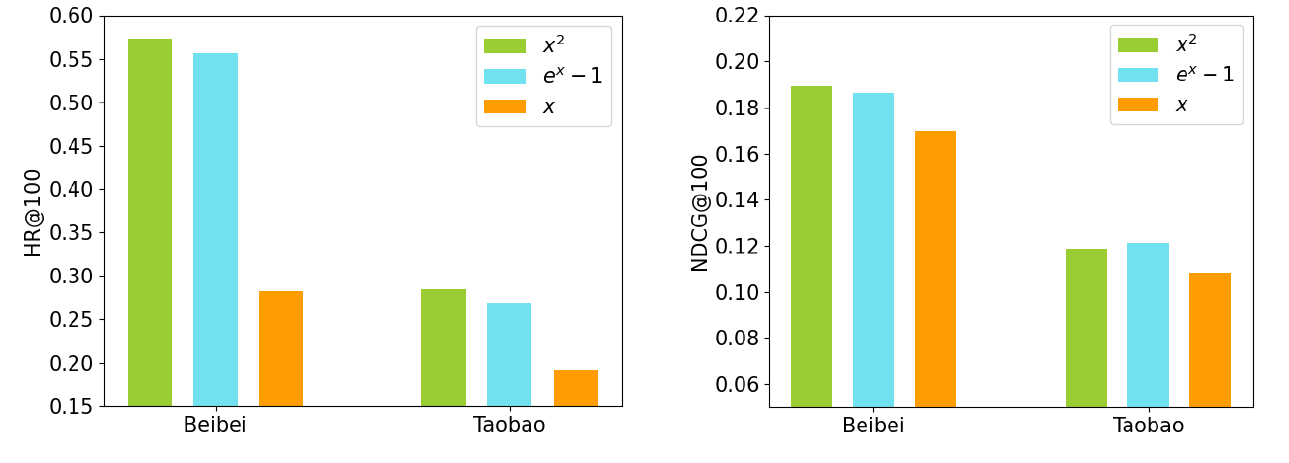}
    \caption{The performance of CHCF with different decorated functions in terms of HR@100 and NDCG@100.} 
    \label{fig:g}
\end{figure}

We first compare multiple decorated functions $g(x) = x , g(x)= x^2, g(x)= e^x -1$ in Equation \ref{equ:CHCF}. Here we append one of the three decorated functions to assess which one can achieve better performance. From Fig. \ref{fig:g}, we find that compared to the other methods, the last two always achieve better results. As we know, $g(x)=x^2$ and $g(x)= e^x -1$ are convex, which leads to a larger gradient in SGD optimization process if the interaction score breaks the criterion seriously. This is consistent with our analysis in the last section. {Furthermore, although $g(x)= e^x -1 $ is not a low-order decorated function, it still has a similar performance to $g(x)= x^2 $ in our implementation. Perhaps the reason is that when the preference score is bounded mostly in practice, $g(x)= e^x -1 $ can still meet the condition of the low-order decorated function.}
For computational efficiency, we set the decorated function to $g(x)= x^2$ as default in this paper. 


\begin{table*}[!h]
\centering
\caption{Performance of variants of CHCF on Taobao.}

\begin{tabular*}{0.99\textwidth}{@{\extracolsep{\fill}}lcccccc@{}}
\toprule

\multirow{2}{*}{\textbf{Taobao}} & \multicolumn{3}{c}{HR} & \multicolumn{3}{c}{NDCG} \\
\cline{2-4}\cline{5-7} 
 & \textbf{10} & \textbf{50} & \textbf{100} & \textbf{10}& \textbf{50 } & \textbf{100}  \\
 \hline
 \textbf{CHCF-O}& 0.0633 & 0.1112 & 0.1525  & 0.0312 & 0.0389 & 0.0457 \\
 
 \textbf{CHCF-I}& 0.1232 & 0.2254 & 0.2687 & 0.0694 & 0.0919 & 0.0985 \\
\textbf{CHCF-U} & 0.1270&0.2294& 0.2769&0.0705&0.0932&0.1009 \\

\textbf{CHCF-H} & 0.0681& 0.1591 & 0.2151  & 0.0486&0.0653&0.0743\\

\hline
\textbf{CHCF-V} & 0.1476 & 0.2110 & 0.2281  & 0.0896 & 0.1016 & 0.1044 \\
\textbf{CHCF-C} & 0.0840 & 0.1444 & 0.1703  & 0.0504 & 0.0637 & 0.0679 \\

\hline
\textbf{CHCF} & \textbf{0.1514} & \textbf{0.2460}  & \textbf{0.2843}  & \textbf{0.0916}  & \textbf{0.1117} &\textbf{0.1189}   \\

\bottomrule
 
\end{tabular*}
\label{tab:ablation_result}
\end{table*}

\subsubsection{Model Ablation and Data Ablation}

Furthermore, we conduct ablation experiments to understand the effectiveness of each of our designs and auxiliary behavior. In particular, we introduce the following variants: 1) CHCF-O. It replaces our loss function with the regression loss in Equation \ref{equ:regress_loss}. For positive interactions, we regress their likelihood to 1 and for negative interactions, we regress their likelihood to 0. With GMF as the basic model, we use different prediction layers ($h_k$) of each behavior, which are randomly initialized and independent.
2) CHCF-H. It replaces the hinge loss with weighted regression loss in CHCF. 
3) CHCF-U. The threshold matrices $S$ and $T$ are only featured by the item heterogeneity matrix $\mathbf{H}$ and all the users share the same threshold for each item-behavior pair.
4) CHCF-I. The threshold matrices $S$ and $T$ are only featured by the user heterogeneity matrix $\mathbf{G}$ and all the items share the same threshold for each user-behavior pair.
5) CHCF-V. It deletes the records of view behavior and merely utilizes the cart records as auxiliary information. 
6) CHCF-C. It deletes the records of cart behavior and merely utilizes the view records as auxiliary information.

We show the performance on dataset Taobao in Table \ref{tab:ablation_result} and the performance on dataset Beibei are similar. Then, we have the following observations:
\R{
\begin{itemize}
    \item CHCF-O performs much worse than both CHCF-H and CHCF, which implies that modeling heterogeneous selection criteria instead of binary regression is able to improve the retrieval performance significantly. It further demonstrates the rationality of our adaptive threshold strategy. 
    \item It can be found that CHCF outperforms CHCF-H greatly. Therefore, our weighted hinge loss can capture different inherent preference strengths, which is capable of enhancing the performance evidently. This clarifies the rationality of hinge loss design in modeling multiple implicit interactions compared to the regression loss. 
    \item CHCF is consistently superior to CHCF-I and CHCF-U, which hence illustrates the importance of modeling selection criteria heterogeneity from both user's and item's views. 
    \item Our full model consistently outperforms CHCF-V and CHCF-C in all settings, demonstrating that our approach enhances purchase prediction via fusing multiple relationships with the MTL framework. Moreover, CHCF-V performs much better than CHCF-C, which indicates that cart behavior is more important for the prediction of purchase interactions. This is in accordance with our intuition. 

\end{itemize}
}

\begin{figure}[!h]
    \centering
    \includegraphics[width=0.7\textwidth]{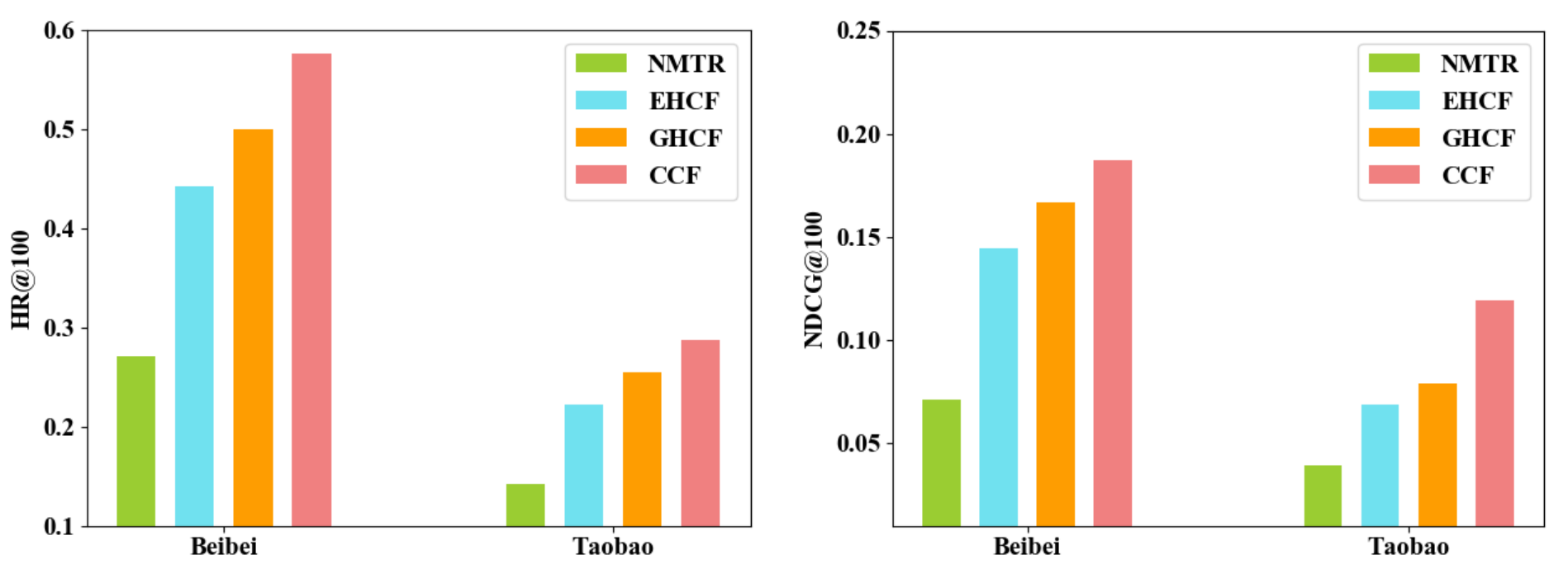}
    \caption{Performance Comparison of Subset with 5$\thicksim$8 Purchase Records.} 
    \label{fig:aul}
\end{figure}
\subsection{Effectiveness Analysis on Sparse Data (RQ4)}

Data sparsity is a significant difficulty for implicit feedback-based recommender systems, and the multi-behavior recommendation is a common solution to this difficulty. In this part, we investigate how the proposed CHCF mitigates the issue when the users have very few records of the target behavior. In particular, we select users having 5$\thicksim$8 purchase records, and 6056 and 11846 users are involved in Beibei and Taobao, respectively. We perform experiments by comparing our CHCF with three baselines NMTR, EHCF and GHCF. 
The results are plotted in Fig. \ref{fig:aul} and we can find that when the user target data grows sparser, our model still keeps a good HR@100 and NDCG@100 performance on Beibei and Taobao, which performs better than the best baseline over 10.0\% consistently. Since CHCF makes use of multiple behavior relationships in a proper manner, it can achieve remarkable performance with sparse interactions of users. Finally, we conclude that the proposed CHCF framework solves the data sparsity problem efficiently to some degree.

\begin{figure}[!h]
\centering
\includegraphics[width=0.7\textwidth]{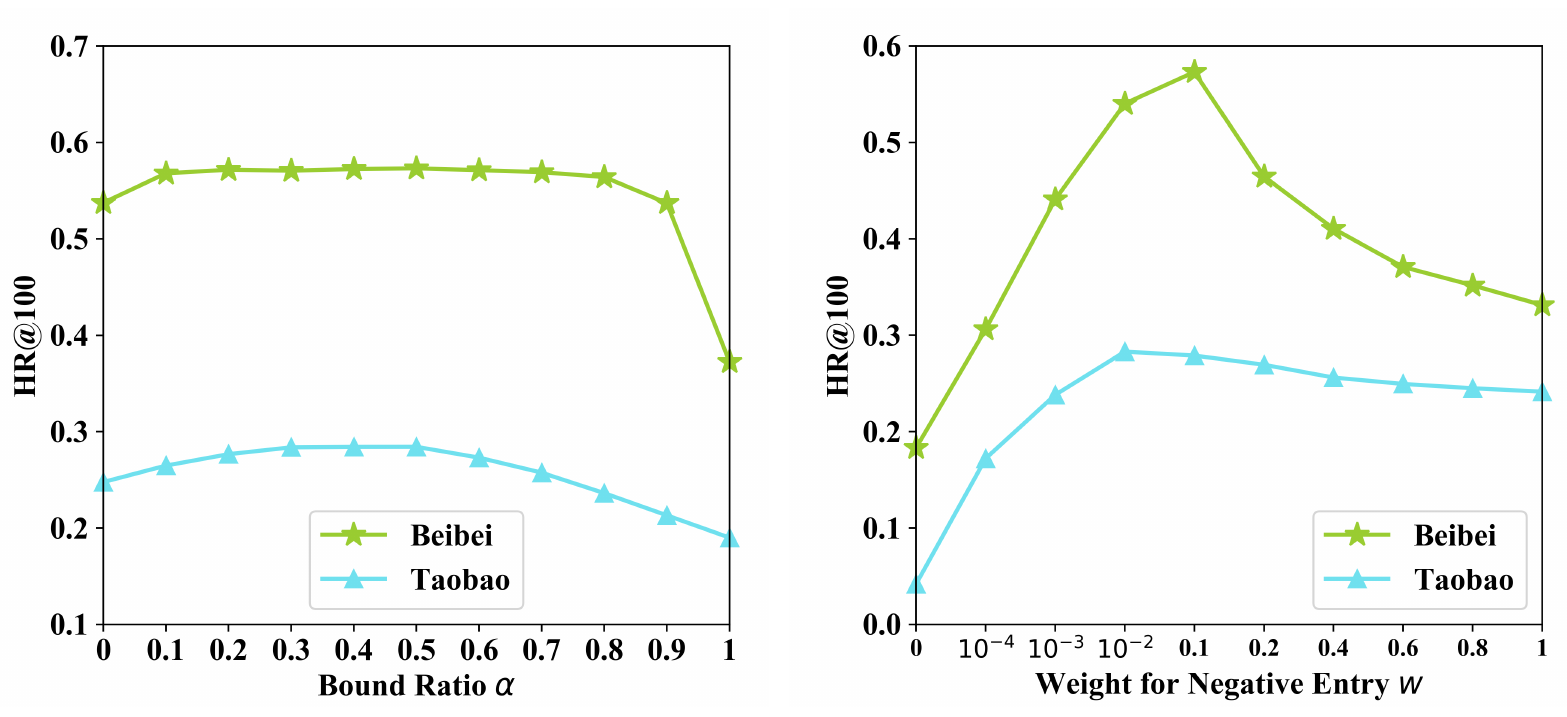}
\caption{Hyper-parameter study of threshold ratio $\alpha$ and weight for negative entry $w$ in terms of HR@100. 
\label{fig:sen}}
\end{figure}
\begin{figure}[!h]
    \centering
    \includegraphics[width=0.7\textwidth]{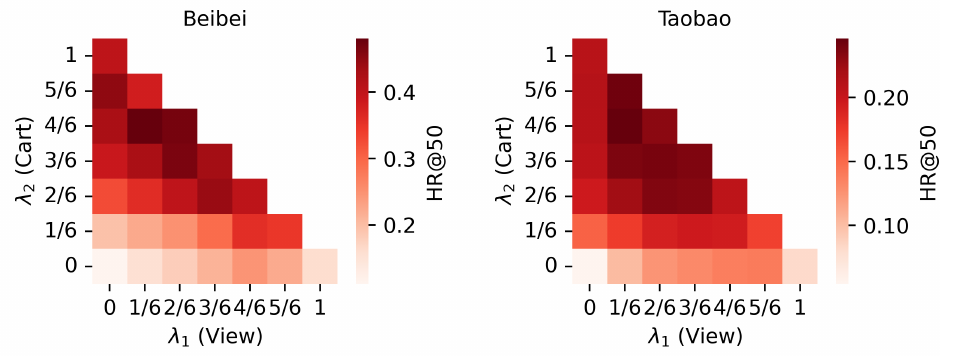}
    \caption{Performance of CHCF on Beibei and Taobao w.r.t. different $\lambda_1$ and $\lambda_2$.} 
    \label{fig:loss_coefficient}
\end{figure}
\subsection{Parameter Sensitivity (RQ5)}\label{sec:hyper} 
In this part, we study the impact of different values of the threshold ratio $\alpha$, different weights for negative entry $w$ and MTL coefficient parameter $\lambda_k$ on our developed CHCF framework. CHCF only introduces a new hyper-parameter threshold ratio $\alpha$, which implies that our model is easy to be tuned compared with other mainstream MTL-based heterogeneous models (i.e., NMTR, EHCF, GHCF). Furthermore, the weight for negative entry $w$ and MTL coefficient parameter $\lambda_k$ are also key hyper-parameters in the whole-data-based learning framework. Thus, we also study their influences on the performance. For each hyper-parameter, we study its influence with other hyper-parameters fixed as default. 

\subsubsection{Impact of Threshold Ratio $\alpha$ } The evaluation on HR@100 of the first two hyper-parameters is shown in Fig. \ref{fig:sen}. For other metrics, the observations are similar. We can observe that the performance is stable when $\alpha$ is around $0.5$ and a too high or too low $\alpha$ may hurt the model performance. The potential reason is that a high $\alpha$ makes thresholds too close to guide user preference learning while a low $\alpha$ leads to learning improper thresholds. 

\subsubsection{Impact of Weight for Negative Entry $w$} In Fig. \ref{fig:sen}, it can be seen that as the weight for negative entry rises from zero to one, the performance of CHCF first increases and then drops. The best results can be obtained at $w=0.1/0.01$ for Beibei/Taobao, demonstrating that the best $w$ is distinct for different datasets. The reason could be that sparseness can be various in various datasets. In conclusion, we need to search for the appropriate $w$ in $[0,1]$ on validation data for each dataset in practice.

\subsubsection{Impact of Weight for MTL Coefficient Parameter $\lambda_k$ } There are three behavior types in both E-commerce datasets (i.e., view, cart and purchase), implying three loss coefficients $\lambda_1$, $\lambda_2$ and $\lambda_3$. Because the sum of three coefficients is 1, if the first two are given, the third one is determined. We vary the first two coefficients in $[0, 1/6, 2/6, 3/6, 4/6, 5/6, 1]$ and the results of HR@50 are shown in Fig. \ref{fig:loss_coefficient}, in which darker blocks mean better performance. It can be found that a relatively large coefficient for the cart behavior performs best on Beibei and Taobao. The potential reason is that purchase interaction is too sparse to offer much information and view interaction is relatively far from the target behavior. For both datasets, the best performance of CHCF is achieved at $[1/6, 4/6, 1/6]$, which not only makes full use of the three types of behavior information and gives more emphasis on the cart behavior. \R{When facing too many behavior types, we suggest to separate these behaviors into two parts, crucial behaviors and non-crucial behaviors. Each part shares the same coefficient $\lambda$. Therefore, there would be only two coefficients required to be tuned. For example, in Tmall, “cart” is selected as the crucial behavior with coefficient $\lambda_1$ while the other three behaviors share the same coefficient $\lambda_2$. Hence, we have $\lambda_1+3\lambda_2=1$ and only vary $\lambda_1$ in $[0,1]$ to search for the best performance. After the tuning process, $\lambda_1$ is set to $1/2$ for “cart” behavior and $\lambda_2$ is set to $1/6$ for the other three behaviors.}

\subsection{\R{Cold-start Problem (RQ6)\label{sec:cold-start}}}

\begin{table*}[t]
\centering
\caption{{Model comparison on solving the cold start problem on Beibei.}}

\begin{tabular*}{0.99\textwidth}{@{\extracolsep{\fill}}lcccccccc@{}}
\toprule

\multirow{3}{*}{\textbf{Beibei}} & \multicolumn{4}{c}{Cold-start Users} & \multicolumn{4}{c}{Cold-start Items} \\
\cline{2-5}\cline{6-9} 
& \multicolumn{2}{c}{HR} & \multicolumn{2}{c}{NDCG} & \multicolumn{2}{c}{HR} & \multicolumn{2}{c}{NDCG} \\
\cline{2-3}\cline{4-5}\cline{6-7}\cline{8-9} 
 & \textbf{50} & \textbf{100} & \textbf{50} & \textbf{100}& \textbf{50} & \textbf{100} & \textbf{50} & \textbf{100}  \\
 \hline
 \textbf{EHCF}& 0.2986 &	0.3662 & 0.0978 & 0.1089 & 0.2938 & 0.4723 & 0.0821 &	0.1108 \\
 

\textbf{MF-CHCF} & \textbf{0.3592} & \textbf{0.4170}  & \textbf{0.1439}  & \textbf{0.1535}  & \textbf{0.3287} &\textbf{0.5684} & \textbf{0.0874} &\textbf{0.1263}   \\
\textbf{GMF-CHCF} & \textbf{0.3724} & \textbf{0.4436}  & \textbf{0.1405}  & \textbf{0.1524}  & \textbf{0.3370} &\textbf{0.5413} & \textbf{0.0846} &\textbf{0.1175}   \\

\bottomrule
 
\end{tabular*}
\label{tab:cold_start_result}
\end{table*}

{In this section, we look into the ability of the proposed CHCF for tackling the cold-start problem. 
When it comes to cold-start users, CHCF can finetune the corresponding user embeddings and user heterogeneity matrix $\mathbf{H}_{cold}$ with online behavior records, while the previous user embeddings and heterogeneity matrix $\mathbf{H}_{previous}$ of previous users, along with the item heterogeneity matrix $\mathbf{G}$ are fixed. When it comes to cold-start items, we can train their embeddings and item heterogeneity matrix $\mathbf{G}_{cold}$ based on behavior records, while the previous item embeddings and item heterogeneity matrix $\mathbf{G}_{previous}$, along with the user heterogeneity matrix $\mathbf{H}$ are fixed. }

{We compare the performance of our CHCF with EHCF for the cold-start problem on Beibei in two settings. In the first settings, we randomly select 1000 users as the cold-start users, and regard the other users as the previous users. Then, the training set can be split into two subsets according to the types of users. We first optimize the model with the training subset of the previous users, and then retrain it with the training subset of the cold-start users, followed by the evaluation on test data of the cold-start users. In the second setting, we randomly select 200 items as the cold-start items, and regard the other items as the previous items. The following process of training and evaluation is similar to the process for cold-start users. The performance ais shown in Table \ref{tab:cold_start_result}. As we can see, in both two settings, MF-CHCF and GMF-CHCF outperform EHCF under all metrics, which demonstrates the superiority of our CHCF for the cold-start problem. Thus, our model is suitable for online recommendation platforms in real-world industry industrial scenarios. 
}

\subsection{Case Study (RQ7)} 
\begin{figure}[t]
\centering
\includegraphics[width=0.7\textwidth]{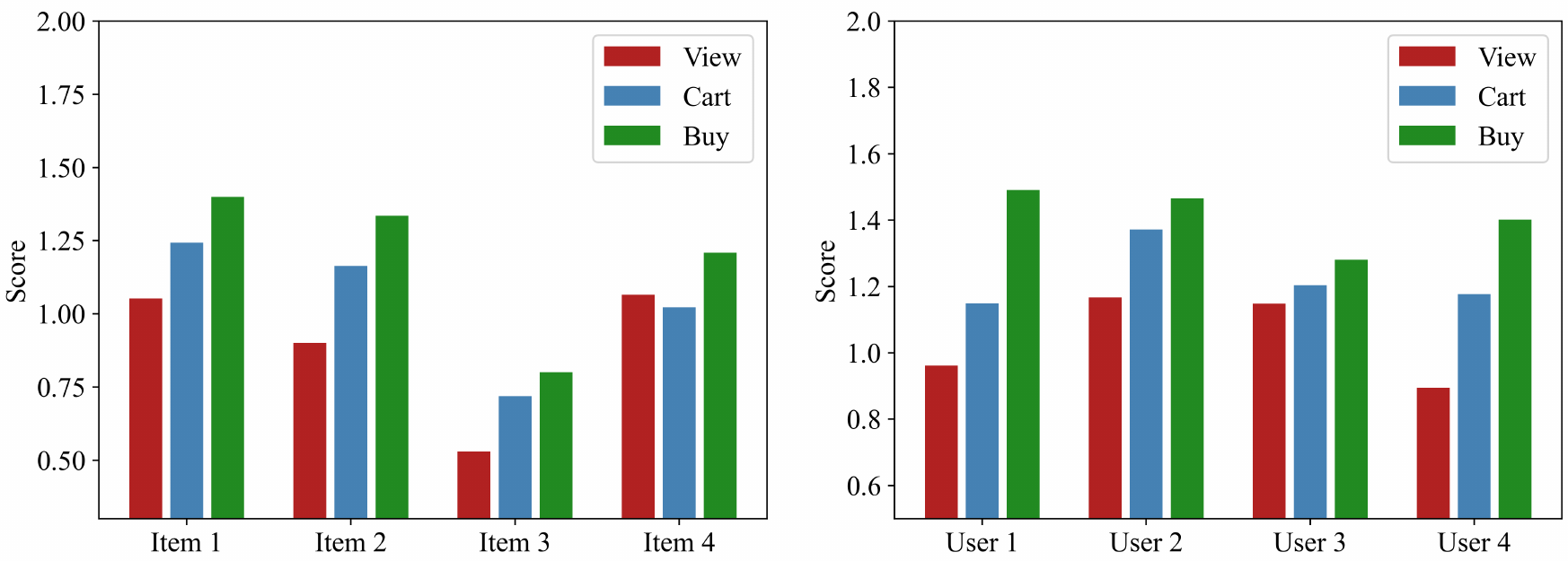}
\caption{Visualization of the learned upper thresholds of several users and items.
\label{fig:dimension}}
\end{figure}

\begin{figure}[t]
\centering
\includegraphics[width=0.7\textwidth]{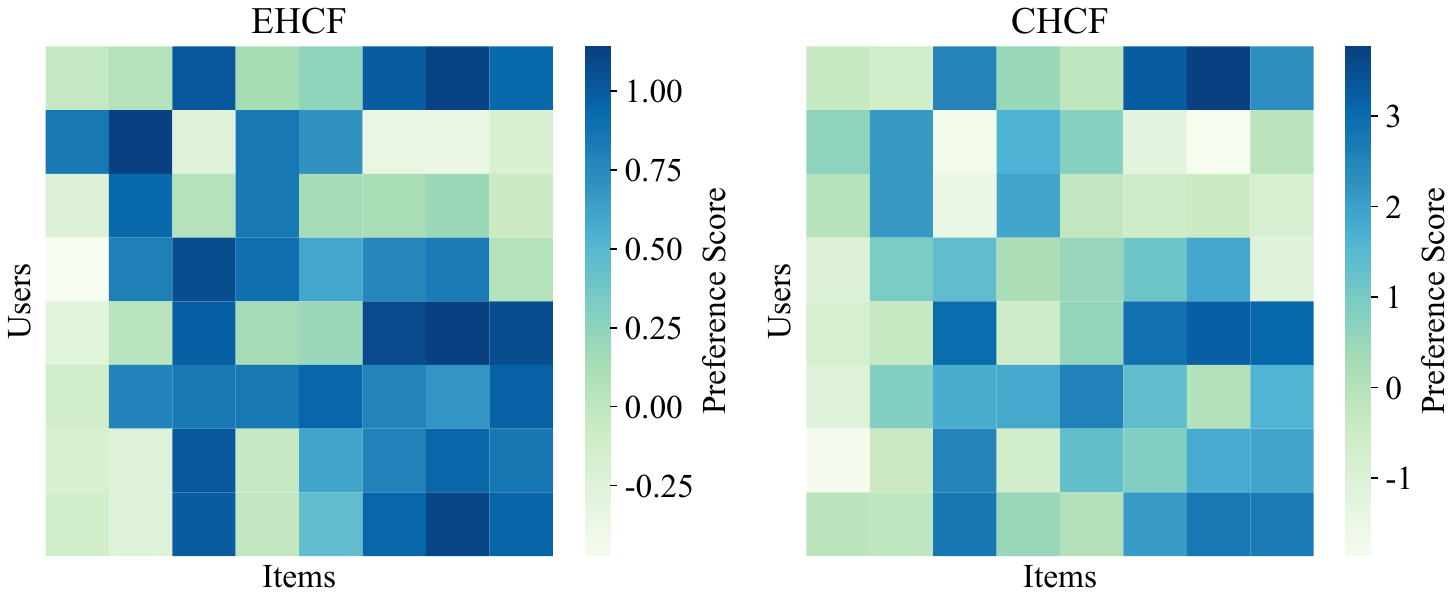}
\caption{{Visualization of the preference scores of several users and items.}
\label{fig:case_inter}}
\end{figure}

To provide a more intuitive understanding of the adaptive user- and item-specific thresholds, we randomly fix a user (item) to visualize thresholds for a list of random items (user) and plot the upper threshold for each behavior in Fig. \ref{fig:dimension}. {Note that we do not add the sigmoid function at the end of collaborative filtering models following~\cite{chen2021graph,chen2020efficient}. Hence, our model could produce preference scores larger than 1, which may result in upper thresholds over 1.} We summarize the following findings: 
1) \R{From the comparison of thresholds among different behaviors, we find that the thresholds of cart behavior are larger than these of view behavior generally, which implies that cart behavior implies a stronger signal. Therefore, without prior information, our CHCF can adaptively learn potential cascading relationships between behaviors from the data.}
2) From the comparison among different users, it is obvious that the thresholds are various for different users, even for the same item on each behavior. As a result, the fixed thresholds for all users are inapposite. Similar results can be obtained for items. Our adaptive threshold strategy could release the impact of the user/item heterogeneity to benefit the performances of recommender systems.
3) In some cases, the thresholds of cart behavior are smaller than these of view behavior, which is interesting and contrary to common sense. This point shows our adaptive threshold strategy is essential for modeling heterogeneous interaction data. {We further visualize the preference scores of several users and items for EHCF (optimized with the regression loss) and CHCF (optimized with our loss) in Fig. \ref{fig:case_inter}. It can be observed that for EHCF, the scores are overall strict close to 1 and 0 for training data which implies potential overfitting, and does not depict the diversities of users and items, especially in multi-behavior scenarios. By contrast, the scores produced by our CHCF are relatively smooth and generalized. }

%% file: conclusion.tex
\section{Conclusion}\label{sec:con}
In this paper, we introduce a novel framework named CHCF for recommendations with heterogeneous user feedback. To be specific, we introduce multiple bounds for different users and items to model selection criteria, which will guide implicit data learning in the form of hinge loss. We further build a unified framework which learns user preference and selection criteria alternatively. Extensive experiments on three popular datasets show that CHCF outperforms the state-of-the-art recommendation models consistently. 
In future work, we will try to incorporate side context knowledge into the procedure of modeling selection criteria, including user’s profiles and items’ textual descriptions. \R{Moreover, we would further attempt to extend our framework into more efficient models for industrial large-scale recommender platforms.}